\documentclass[aps,prb, floatfix,notitlepage,nofootinbib,twocolumn,superscriptaddress,reprint,preprintnumbers]{revtex4-2}
\usepackage{amsmath}
\usepackage{amssymb}
\usepackage{graphicx}
\usepackage[tight]{subfigure}
\usepackage{enumitem}
\usepackage{tikz}
\usetikzlibrary{calc,fadings,decorations.pathreplacing,shapes,shapes.multipart,arrows,shapes.misc,intersections,positioning}
\usepackage{wrapfig}
\usepackage{bm}
\usepackage[autostyle, english = american]{csquotes}
\usepackage{xifthen}
\usepackage{xargs}
\usepackage[tight]{subfigure}
\MakeOuterQuote{"}
\usepackage{pifont}
\usepackage{array}
\usepackage{multirow}

\usepackage{amsthm}
\usepackage{physics}

\theoremstyle{definition}

\theoremstyle{theorem}
\newtheorem{theorem}{Theorem}

\theoremstyle{proposition}

\theoremstyle{lemma}
\newtheorem{lemma}{Lemma}

\theoremstyle{claim}
\newtheorem*{claim}{Claim}

\theoremstyle{corollary}
\newtheorem*{corollary}{Corollary}

\usepackage[english]{babel}
\makeatletter
\def\bbl@set@language#1{%
  \edef\languagename{%
    \ifnum\escapechar=\expandafter`\string#1\@empty
    \else\string#1\@empty\fi}%
  \@ifundefined{babel@language@alias@\languagename}{}{%
    \edef\languagename{\@nameuse{babel@language@alias@\languagename}}%
  }%
  \select@language{\languagename}%
  \expandafter\ifx\csname date\languagename\endcsname\relax\else
    \if@filesw
      \protected@write\@auxout{}{\string\select@language{\languagename}}%
      \bbl@for\bbl@tempa\BabelContentsFiles{%
        \addtocontents{\bbl@tempa}{\xstring\select@language{\languagename}}}%
      \bbl@usehooks{write}{}%
    \fi
  \fi}
\newcommand{\DeclareLanguageAlias}[2]{%
  \global\@namedef{babel@language@alias@#1}{#2}%
}
\makeatother
\DeclareLanguageAlias{en}{english}

\newcolumntype{C}{>{$\displaystyle} c <{$}}

\newcommand{\ie}{\textit{i.e.} }

\newcommand{\eps}{\epsilon}
\newcommand{\ot}{\otimes}
\newcommand{\ra}{\rightarrow}

\newcommand{\cO}{\mathcal{O}}

\newcommand{\I}{\mathbb{I}}

\usepackage{bm}

\newcommand\appref[1]{App.~\ref{#1}}

\usepackage[colorlinks=true,citecolor=blue,linkcolor=magenta]{hyperref}
\usepackage{cleveref}

\graphicspath{{./}{./images/}}

\usepackage{xifthen}
\usepackage{xargs}

\newcommandx{\drawbox}[6][1=0,2=0,3=1,4=1,5=,6=]{
  \ifthenelse{\equal{#5}{}}{
    \draw[line width = 0.5pt] (#1,#2) rectangle (#3,#4);
  }{
    \draw[line width = 0.5pt, fill = #5] (#1,#2) rectangle (#3,#4);
  }
  \node () at (#1*0.5+#3*0.5,#2*0.5+#4*0.5) {#6};
}

\newcommandx{\regbox}[4][1=0,2=0,3=,4=]{
  \begin{scope}[shift={(#1,#2)}]
    \drawbox[0][0][1.5][1][#3][#4];
  \end{scope}
}

\newcommandx{\regboxt}[8][1=0,2=0,3=,4=,5=,6=,7=,8=]{
  \begin{scope}[shift={(#1,#2)}]
    \drawbox[0][0][1.5][1][#3][#4];
    \node[above] () at (0,1) {#5};
    \node[above] () at (1.5,1) {#6};
    \node[below] () at (1.5,0) {#7};
    \node[below] () at (0,0) {#8};
  \end{scope}
}

\newcommandx{\sqzbox}[5][1=0,2=0,3=,4=,5=]{
  \begin{scope}[shift={(#1,#2)}]
    \drawbox[0][0][2.25][1][#3][#4];
    \ifthenelse{\equal{#5}{s}}{
      \draw (0,0.1)--++(-0.1,0)--++(0,1)--++(2.25,0)--++(0,-0.1);
    }{}
  \end{scope}
}

\newcommandx{\ballon}[5][1=0,2=0,3=0.25,4=0.5,5=]{
  \begin{scope}[shift={(#1,#2)}]
    \pgfmathsetmacro{\r}{#4};
    \pgfmathsetmacro{\cent}{#3+#4};
    \draw (0,0)--++(0,#3);
    \draw (0,\cent) circle (\r);
    \draw (0,\cent+\r)--++(0,#3);
    \ifthenelse{\equal{#5}{}}{}{
      \node () at (0,\cent) {#5};
    }
  \end{scope}
}

\newcommandx{\wallebox}[5][1=,2=,3=,4=,5=]
{
  \fineq[-0.8ex][0.8][0.8]{
    \sqzbox[0][0];
    \ifthenelse{\equal{#1}{t}}{
      \ballon[0.5][1][0.25][0.5][$#2$];
      \ballon[1.75][1][0.25][0.5][$#3$];
    }{}
    \ifthenelse{\equal{#1}{b}}{
      \ballon[0.5][-1][0.25][0.5][$#4$];
      \ballon[1.75][-1][0.25][0.5][$#5$];
    }{}
    \ifthenelse{\equal{#1}{tb}}{
      \ballon[0.5][1][0.25][0.5][$#2$];
      \ballon[1.75][1][0.25][0.5][$#3$];
      \ballon[0.5][-1.5][0.25][0.5][$#4$];
      \ballon[1.75][-1.5][0.25][0.5][$#5$];
    }{}
  }
}

\newcommandx{\sufourpart}[5][1=0,2=0,3=0.25,4=0.5,5=]{
  \begin{scope}[shift={(#1,#2)}]
    \pgfmathsetmacro{\r}{#4};
    \pgfmathsetmacro{\cent}{#3+#4};
    \draw (0,0)--++(0,#3);
    \draw (-\r,\cent-\r) rectangle (\r,\cent+\r);
    \draw (0,\cent+\r)--++(0,#3);
    \ifthenelse{\equal{#5}{}}{}{
      \node () at (0,\cent) {#5};
    }
  \end{scope}
}

\newcommandx{\sufour}[5][1=,2=,3=,4=,5=]
{
  \fineq[-0.8ex][0.6][1]{
    \sqzbox[0][0][][$#1$];
    \sufourpart[0.5][1][0.25][0.5][$#2$];
    \sufourpart[1.75][1][0.25][0.5][$#3$];
    \sufourpart[0.5][-1.5][0.25][0.5][$#4$];
    \sufourpart[1.75][-1.5][0.25][0.5][$#5$];
  }
}

\newcommandx{\squbox}[4][1=0,2=0,3=,4=]{
  \begin{scope}[shift={(#1,#2)}]
    \drawbox[0][0][1][1][#3][#4];
  \end{scope}
}

\newcommandx{\uonesite}[7][1=0,2=0,3=,4=,5=,6=,7=]{
  \begin{scope}[shift={(#1,#2)}]
    \draw[line width = 0.5pt] (0.5,0.25)--++(0,-0.25) coordinate (A);
    \draw[line width = 0.5pt] (0.5,1.25)--++(0,0.25) coordinate (B);
    \ifthenelse{\equal{#7}{s}}{
      \squbox[0][0.25][#3][#4];
    }{
      \regbox[-0.25][0.25][#3][#4];      
    }
    \ifthenelse{\equal{#5}{}}{}{
      \node[below] () at (A) {#5};
    }
    \ifthenelse{\equal{#6}{}}{}{
      \node[above] () at (B) {#6};
    }
  \end{scope}
}

\newcommandx{\ugate}[5][1=0,2=0,3=,4=,5=]{
  \begin{scope}[shift={(#1,#2)}]
      \draw[line width = 0.5pt] (0,0)--++(0,0.25);
      \draw[line width = 0.5pt] (0,1.25)--++(0,0.25);
      \draw[line width = 0.5pt] (1,0)--++(0,0.25);
      \draw[line width = 0.5pt] (1,1.25)--++(0,0.25);
      \regbox[-0.25][0.25][#3][#4];
      \ifthenelse{\equal{#5}{p}}{
        \draw[line width = 2pt] (-0.1,0)--(1.1,0);
        \draw (0,0)--++(0,-0.1);
        \draw (1,0)--++(0,-0.1);
      }{}
  \end{scope}
}

\newcommandx{\dualgate}[8][1=0,2=0,3=,4=,5=,6=,7=,8=]{
  \begin{scope}[shift={(#1,#2)}]
    \ifthenelse{\equal{#3}{l} \OR \equal{#3}{}}{
      \draw (0,-0.2)--(0.2,-0.2)--(0.2,0.2)--(0,0.2);
      \draw (0.2,0.2)--(0.5,0.5);
      \node () at (0.6,0.6) {$#5$};
      \draw (0.2,-0.2)--(0.5,-0.5);
      \node () at (0.6,-0.6) {$#7$};
      \node () at (0,0) {$#8$};
    }{}
    \ifthenelse{\equal{#3}{r} \OR \equal{#3}{}}{
      \draw (0,-0.2)--(-0.2,-0.2)--(-0.2,0.2)--(0,0.2);
      \draw (-0.2,0.2)--(-0.5,0.5);
      \node () at (-0.6,0.6) {$#4$};
      \draw (-0.2,-0.2)--(-0.5,-0.5);
      \node () at (-0.6,-0.6) {$#6$};
      \node () at (0,0) {$#8$};
    }{}
  \end{scope}
}

\newcommandx{\ugater}[3][1=0,2=0,3=]{
  \begin{scope}[shift={(#1,#2)}]
      \draw[line width = 0.5pt] (1,0)--++(0,0.25);
      \draw[line width = 0.5pt] (1,1.25)--++(0,0.25);
      \ifthenelse{\equal{#3}{}}{
        \draw[line width = 0.5pt] (0.75,0.25)--++(0.5,0)--++(0,1)--++(-0.5,0);
      }{
        \draw[line width = 0.5pt, fill=#3] (0.75,0.25)--++(0.5,0)--++(0,1)--++(-0.5,0);
      }
  \end{scope}
}

\newcommandx{\ugatel}[3][1=0,2=0,3=]{
  \begin{scope}[shift={(#1,#2)}]
      \draw[line width = 0.5pt] (0,0)--++(0,0.25);
      \draw[line width = 0.5pt] (0,1.25)--++(0,0.25);
      \ifthenelse{\equal{#3}{}}{
        \draw[line width = 0.5pt] (0.25,0.25)--++(-0.5,0)--++(0,1)--++(0.5,0);
      }{
        \draw[line width = 0.5pt, fill=#3] (0.25,0.25)--++(-0.5,0)--++(0,1)--++(0.5,0);
      }
  \end{scope}
}

\newcommandx{\shortarc}[4][1=0,2=0,3=,4=]{
  \begin{scope}[shift={(#1,#2)}]
    \ifthenelse{\equal{#3}{r}}{
      \draw[line width = 0.5pt] (0,0) to[out=-90,in=0] (-0.1+0.06,0.1-0.18) to[out=180,in=-90] (-0.1,0.1);
      \pgfmathsetmacro{\flag}{-1};
    }{
      \draw[line width = 0.5pt] (0,0) to[out=90,in=0] (-0.06,0.18) to[out=180,in=90] (-0.1,0.1);
      \pgfmathsetmacro{\flag}{1};
    }
    \ifthenelse{\equal{#4}{l}}{
      \draw (0,0)--++(0,-0.2);
      \draw (-0.1,0.1)--++(0,-0.2);
    }{}
  \end{scope}
}

\newcommandx{\longarc}[4][1=0,2=0,3=,4=]{
  \begin{scope}[shift={(#1,#2)}]
    \ifthenelse{\equal{#3}{r}}{
      \draw[line width = 0.5pt] (0,0) to[out=-90,in=0] (-0.1*3+0.06*3,0.1*3-0.14*3) to[out=180,in=-90] (-0.1*3,0.1*3);
      \pgfmathsetmacro{\flag}{-1};
    }{
      \draw[line width = 0.5pt] (0,0) to[out=90,in=0] (-0.06*3,0.14*3) to[out=180,in=90] (-0.1*3,0.1*3);
      \pgfmathsetmacro{\flag}{1};
    }
    \ifthenelse{\equal{#4}{l}}{
      \draw (0,0)--++(0,-0.2);
      \draw (-0.1*3,0.1*3)--++(0,-0.2);
    }{}
  \end{scope}
}

\newcommandx{\idarc}[4][1=0,2=0,3=,4=]{
  \begin{scope}[shift={(#1,#2)}]
    \shortarc[0][0][#3][#4];
    \shortarc[-0.2][0.2][#3][#4];
  \end{scope}
}

\newcommandx{\swaparc}[4][1=0,2=0,3=,4=]{
  \begin{scope}[shift={(#1,#2)}]
    \longarc[0][0][#3][#4];
    \shortarc[-0.1][0.1][#3][#4];
  \end{scope}
}

\newcommandx{\idst}[3][1=0,2=0,3=]{
  \ifthenelse{\equal{#3}{r}}{
    \pgfmathsetmacro{\flag}{-1};
  }{
    \pgfmathsetmacro{\flag}{1};
  }
  \begin{scope}[shift={(#1,#2)}]
    \draw[line width = 0.5pt] (0,0) to[out=\flag*90,in=180] (0.15,\flag*0.2) to[out=0,in=\flag*90] (0.3,0);
    \draw[line width = 0.5pt] (0.4,0) to[out=\flag*90,in=180] (0.4+0.15,\flag*0.2) to[out=0,in=\flag*90] (0.7,0);
  \end{scope}
}

\newcommandx{\swapst}[3][1=0,2=0,3=]{
  \ifthenelse{\equal{#3}{r}}{
    \pgfmathsetmacro{\flag}{-1};
  }{
    \pgfmathsetmacro{\flag}{1};
  }
  \begin{scope}[shift={(#1,#2)}]
      \draw[line width = 0.5pt] (0,0) to[out=\flag*90,in=180] (0.4,\flag*0.25) to[out=0,in=\flag*90] (0.8,0);
      \draw[line width = 0.5pt] (0.3,0) to[out=\flag*90,in=180] (0.4,\flag*0.15) to[out=0,in=\flag*90] (0.5,0);
  \end{scope}
}

\newcommandx{\uonesitestack}[4][1=,2=,3=1,4=]
{
  \begin{scope}[shift={(#1,#2)}]
    \foreach \x in {#3,...,1}{
      \pgfmathsetmacro{\shiftx}{-0.1*(2*\x-2)};
      \pgfmathsetmacro{\shifty}{0.1*(2*\x-2)};
      \uonesite[\shiftx-0.1][\shifty+0.1][blue!50];
      \uonesite[\shiftx][\shifty][red!50];
    }
    \ifthenelse{\equal{#4}{id}}{
      \idarc[0.5][1.5];
    }{}
    \ifthenelse{\equal{#4}{swap}}{
      \swaparc[0.5][1.5];
    }{}
  \end{scope}
}

\newcommandx{\ugatestack}[5][1=0,2=0,3=1,4=,5=]
{
  \begin{scope}[shift={(#1,#2)}]
    \foreach \x in {#3,...,1}{
      \pgfmathsetmacro{\shiftx}{-0.1*(2*\x-2)};
      \pgfmathsetmacro{\shifty}{0.1*(2*\x-2)};
      \ugate[\shiftx-0.1][\shifty+0.1][blue!50];
      \ugate[\shiftx][\shifty][red!50];
    }
    \ifthenelse{\equal{#4}{id}}{
      \idarc[0][1.5];
    }{}
    \ifthenelse{\equal{#4}{swap}}{
      \swaparc[0][1.5];
    }{}
    \ifthenelse{\equal{#5}{id}}{
      \idarc[1][1.5];
    }{}
    \ifthenelse{\equal{#5}{swap}}{
      \swaparc[1][1.5];
    }{}
  \end{scope}
}

\newcommandx{\mpsa}[7][1=0,2=0,3=,4=,5=,6=,7=]
{
  \begin{scope}[shift={(#1,#2)}]
    \ifthenelse{\equal{#3}{}}{
      \draw (0,0) circle (0.5);
    }{
      \draw[fill=#3] (0,0) circle (0.5);
    }
    \node () at (0,0) {#4};
    \draw (-0.5,0)--++(-0.5,0);
    \draw (0.5,0)--++(0.5,0);
    \draw (0,0.5)--++(0,0.5);
    \node () at (-1.5,0) {#5};
    \node () at (0,1.5) {#6};
    \node () at (1.5,0) {#7};
  \end{scope}
}

\newcommandx{\mpsastack}[7][1=0,2=0,3=,4=,5=,6=,7=]
{
  \begin{scope}[shift={(#1,#2)}]
    \mpsa[-0.3][0.3][blue!50][#4][#5][#6][#7];
    \mpsa[-0.2][0.2][red!50][#4][#5][#6][#7];
    \mpsa[-0.1][0.1][blue!50][#4][#5][#6][#7];
    \mpsa[0][0][red!50][#4][#5][#6][#7];
  \end{scope}
}

\newcommandx{\fineq}[4][1=-.8ex,2=1,3=1]{
  \begin{tikzpicture}[baseline={([yshift=#1]current  bounding  box.center)}, scale = #2, every node/.style={scale = #3}]
    #4
  \end{tikzpicture}
}

\newcommandx{\ideq}[1][1=]{
  \fineq{
    \idst[0][0][#1]
  }
}
\newcommandx{\swapeq}[1][1=]{
  \fineq{
    \swapst[0][0][#1]
  }
}

\newcommandx{\idket}[0]{
  |
  \fineq[-0.6ex]{
    \idst[0][0][r]
  } \rangle 
}

\newcommandx{\swapket}[0]{
  |
  \fineq[-0.4ex]{
    \swapst[0][0][r]
  } \rangle 
}

\newcommandx{\idbra}[0]{
  \langle 
  \fineq{
    \idst[0][0][]
  } |
}

\newcommandx{\swapbra}[0]{
  \langle 
  \fineq{
    \swapst[0][0][]
  } |
}

\newcommandx{\uonesiteeq}[3][1=,2=,3=]
{
  \fineq{
    \uonesite[0][0][][#1][#2][#3];
    \node () at (-0.5,0.75) {};
    \node () at (1.5,0.75) {};
  }
}

\newcommandx{\ugateeq}[3][1=,2=,3=]
{
  \fineq{
    \ugate[0][0][][#1][#2][#3];
    \node () at (-0.5,0.75) {};
    \node () at (1.5,0.75) {};
  }
}

\newcommandx{\uonesitetwoeq}[4][1=,2=,3=,4=]
{
  \fineq{
    \uonesite[0][0][][#1][#3];
    \uonesite[0][1.5][][#2][#4];
    \node () at (-0.5,0.75) {};
    \node () at (1.5,0.75) {};
  }
}

\newcommandx{\hfbox}[5][1=0,2=0,3=,4=l,5=]{
  \ifthenelse{\equal{#4}{r}}{
    \pgfmathsetmacro{\flag}{-1};
  }{
    \pgfmathsetmacro{\flag}{1};
  }
  \begin{scope}[shift={(#1,#2)}]
    \draw (0,0) --++ (\flag*1,0) --++ (0,1) --++(-\flag*1,0);
    \node () at (\flag*0.5,0.5) {#3};
    \ifthenelse{\equal{#5}{s}}{
      \ifthenelse{\equal{#4}{l}}{
        \draw (0.9,1)--++(0,0.1) --++ (-1,0);
      }{}
      \ifthenelse{\equal{#4}{r}}{
        \draw (-1,0.1)--++(-0.1,0)--++(0,1) --++ (1,0);
      }{}
    }{}
  \end{scope}
}

\newcommandx{\tribox}[4][1=,2=,3=,4=]{
  \fineq[-0.8ex][0.5][0.8]{
    \hfbox[0][1.25][$#1$][l][#4];
    \hfbox[2.25][1.25][$#2$][r][#4];
    \sqzbox[0][0][][$#3$][#4];
    \node () at (-0.25,1.25) {};
    \node () at (2.5,1.25) {};
  }
}

\newcommandx{\lshapebox}[4][1=,2=,3=l,4=]{
  \fineq[-0.8ex][0.5][0.8]{
    \ifthenelse{\equal{#3}{l}}{
      \hfbox[0][1.25][$#1$][l][#4];
    }{
      \hfbox[2.25][1.25][$#1$][r][#4];
    }
    \sqzbox[0][0][][$#2$][#4];
    \node () at (-0.25,1.25) {};
    \node () at (2.5,1.25) {};
  }
}

\newcommandx{\hexbox}[8][1=,2=,3=,4=,5=,6=,7=,8=]{
  \fineq[-0.8ex][0.5][0.8]{
    \sqzbox[0][2.5][][$#1$][#8];
    \sqzbox[2.5][2.5][][$#2$][#8];
    \hfbox[0][1.25][$#3$][l][#8];
    \sqzbox[1.25][1.25][][$#4$][#8];
    \hfbox[2.5+2.25][1.25][$#5$][r][#8];
    \sqzbox[0][0][][$#6$][#8];
    \sqzbox[2.5][0][][$#7$][#8];
    \node () at (-0.25,1.25) {};
    \node () at (5,1.25) {};
  }
}

\newcommandx{\uucontract}[5][1=,2=,3=,4=,5=]{
  \fineq[-0.8ex][0.5][0.8]{
    \regbox[-0.35][0.35][blue!50];
    \ifthenelse{\equal{#1}{tr}}{
      \idline[0][1.25][r]
    }{
      \idline[0][1.25]
    }
    \ifthenelse{\equal{#2}{tr}}{
      \idline[1][1.25][r]
    }{
      \idline[1][1.25]
    }
    \ifthenelse{\equal{#3}{tr}}{
      \idline[0][0.25][u]
    }{}
    \ifthenelse{\equal{#3}{epr}}{
      \idline[0][0.25][epr]
    }{}
    \ifthenelse{\equal{#3}{}}{
      \idline[0][0]
    }{}
    
    \ifthenelse{\equal{#4}{tr}}{
      \idline[1][0.25][u]
    }{}
    \ifthenelse{\equal{#4}{epr}}{
      \idline[2][0.25][epr]
    }{}
    \ifthenelse{\equal{#4}{eprtr}}{
      \idline[1][0.25][eprtr]
    }{}
    \ifthenelse{\equal{#4}{}}{
      \idline[1][0]
    }{}
    \regbox[-0.25][0.25][red!50];
    \ifthenelse{\equal{#5}{}}{
    }{
      \node () at (0.5,0.75) {#5};
    }
  }
}

\newcommandx{\idline}[3][1=0,2=0,3=]{
  \begin{scope}[shift={(#1,#2)}]
    \ifthenelse{\equal{#3}{r}}{
      \begin{scope}[shift={(0,0.25)}]
        \draw[line width = 0.5pt] (0,-0.25)-- (0,0) to[out=90,in=0] (-0.06,0.18) to[out=180,in=90] (-0.1,0.1) --++ (0,-0.25);
      \end{scope} 
    }{}   
    \ifthenelse{\equal{#3}{u}}{
      \begin{scope}[shift={(0,-0.25)}]
      \draw[line width = 0.5pt] (0,0.25)--(0,0) to[out=-90,in=0] (-0.1+0.06,0.1-0.18) to[out=180,in=-90] (-0.1,0.1)--++( 0,0.25);
      \end{scope} 
    }{}
    \ifthenelse{\equal{#3}{epr}}{
      \draw[line width = 0.5pt] (0,0) -- (0,-0.2) to[out=-90,in=0] (-0.5,-0.55) to[out=180,in=-90] (-1,-0.3)--(-1,0.0);
      \begin{scope}[shift={(-0.1,0.1)}]
        \draw[line width = 0.5pt] (0,0) -- (0,-0.25) to[out=-90,in=0] (-0.5,-0.55) to[out=180,in=-90] (-1,-0.15)--(-1,0.0);
      \end{scope}  
    }{}
    \ifthenelse{\equal{#3}{eprtr}}{
      \draw[line width = 0.5pt] (0,0) -- (0,-0.3) to[out=-90,in=180] (0.5,-0.55) to[out=0,in=-90] (1,-0.2)--(1,0) to[out=90,in=0] (1-0.06,0.18) to[out=180,in=90] (1-0.1,0.1);
      \begin{scope}[shift={(-0.1,0.1)}]
        \draw[line width = 0.5pt] (0,0) -- (0,-0.15) to[out=-90,in=180] (0.5,-0.55) to[out=0,in=-90] (1,-0.25)--(1,0.0);
      \end{scope}  
    }{}
    \ifthenelse{\equal{#3}{}}{
      \draw[line width = 0.5pt] (0,0.25)--(0,0);
      \draw[line width = 0.5pt] (-0.1,0.1)--++(0,0.25);
    }{}
  \end{scope}
}

\newcommandx{\idlineu}[2][1=0,2=0]{
  \fineq[-0.8ex][0.75][0.8]{\idline[0][0][u]} 
}

\newcommandx{\idlineepr}[2][1=0,2=0]{
  \fineq[-0.8ex][0.5][0.8]{\idline[0][0][epr]} 
}

\begin{document}

\title{Maximal entanglement velocity implies dual unitarity}

\author{Tianci Zhou}
\affiliation{Center for Theoretical Physics, Massachusetts Institute of Technology, Cambridge, Massachusetts 02139, USA}
\author{Aram W. Harrow}
\affiliation{Center for Theoretical Physics, Massachusetts Institute of Technology, Cambridge, Massachusetts 02139, USA}
\date{\today}

\begin{abstract}
  A global quantum quench can be modeled by a quantum circuit with
  local unitary gates. In general, entanglement grows linearly at a
  rate given by entanglement velocity, which is upper bounded by the
  growth of the light cone.

  We show that the unitary interactions achieving the maximal rate
  must remain unitary if we exchange the space and time directions --
  a property known as dual unitarity. Our results are robust:
  approximate maximal entanglement velocity also implies approximate
  dual unitarity.

  We further show that maximal entanglement velocity is always
  accompanied by a specific dynamical pattern of entanglement, which
  yields simpler analyses of several known exactly solvable models.
  
\end{abstract}
\preprint{MIT-CTP/5368}
\maketitle

{\bf Introduction.} 
The propagation of information can never exceed the speed of light, due to Lorentz invariance.  Any particle actually achieving this speed must be massless, and lower speed limits can be placed on massive particles when energy is limited. In non-relativistic systems where the speed of light is effectively infinite, the locality of the interactions poses an emergent constraint\cite{lieb_finite_1972}. In this letter, we study the speed limit of entanglement -- a measure of quantum information -- in locally interacting quantum circuit.  As with the speed of light, it will turn out that local unitary interactions (or ``gates'') that achieve the maximum velocity of spreading entanglement have a special form.

There is a natural notion of entanglement velocity in a global quantum quench\cite{casini_spread_2015,liu_entanglement_2014,hartman_speed_2015}. When a short-range entangled state $|\psi_0\rangle$ is unitarily evolved, in general a (small) subsystem $Q$ will thermalize. After sufficiently long time, the entanglement (or von Neumann) entropy $S(Q)$ of the subsystem $Q$ will saturate to its equilibrium value. To set the stage, we consider an infinite lattice qudit system in one dimension with local Hilbert space dimension $q$, and take a semi-infinite region $Q$ as the subsystem. We assume that the unitary evolution can thermalize the state $|\psi_0 \rangle$ to infinite temperature. On the way to equilibrium, the von Neumann entropy of $Q$ typically grows linearly in $t$\cite{calabrese_quantum_2007,calabrese_quantum_2016-1,kim_ballistic_2013}
\begin{equation}
S(Q)_t \equiv S(Q)_{\rho(t)} \equiv
 - \tr( \rho_Q \ln \rho_Q ) \sim \ln (q)  v_E t . 
\end{equation}
The linear coefficient divided by the entropy density $\ln(q)$ has the dimension of velocity. It is thus called the entanglement velocity and denoted as $v_E$. A more precise definition of $v_E$ is the asymptotic growth rate (maximized over short-range initial states)
\begin{equation}
\label{eq:vE_def}
v_E = \lim_{t \rightarrow \infty} \frac{S(Q)_t}{t \ln(q) }. 
\end{equation}

We model spatially local interaction by a quantum circuit with local gates in a brickwork structure (Fig.~\ref{fig:circuit}). The brickwork unitary circuit has been extensively studied in recent research about quantum chaos\cite{nahum_operator_2018,von_keyserlingk_operator_2018,rakovszky_diffusive_2018,khemani_operator_2018} and entanglement\cite{nahum_quantum_2017,bertini_exact_2019-1,von_keyserlingk_operator_2018,napp_efficient_2020,chan_solution_2018}, bearing fruitful results. Here we assume the gates in the circuit are identical, unless explicitly stated otherwise. Taking the depth as time, the construction has a natural light cone velocity $v_{\rm LC} = 1$, so that effective system size is at most $2v_{\rm LC} t=2t$, with $ t$ sites each for $Q$ and its complement.  Since $S(Q)\leq \ln(\dim Q) = t\ln(q)$, implying that $v_E \le 1$.


In the study of quantum chaos, researchers discovered certain (generally non-integrable) brickwork circuits whose $v_E$ is exactly $1$\cite{bertini_entanglement_2018}. The gate is taken to be {\it self dual} as we now define. We denote a two-site unitary gate as $u$ with element $u_{ij,kl}$. By definition, we obtain an identity matrix when multiplying $u$ with its Hermitian conjugate, \ie $u_{ij,kl} u^*_{i'j',kl} = \delta_{ii'} \delta_{jj'}$. We draw this unitarity relation as
\begin{equation}
\label{eq:u-cond}
\fineq[-0.8ex][0.75][0.8]{
    \regbox[-0.35][0.35][blue!50];
    \idline[0][1.25];
    \idline[1][1.25];
    \idline[0][0.25][u];
    \idline[1][0.25][u];
    \regbox[-0.25][0.25][red!50];
    \node () at (0.15,1.5+0.1) {$i$};
    \node () at (-0.2,1.5+0.2) {$i'$};
    \node () at (1.15,1.5+0.1) {$j$};
    \node () at (1-0.2,1.5+0.2) {$j'$};
    \node () at (0.05,-0.25) {$k$};
    \node () at (1+0.05,-0.25) {$l$};
}
 = 
\fineq[-0.8ex][1][0.8]{
\idline[0][0][u]
\node () at (0.1,0.1) {$i$};
\node () at (-0.2,0.2) {$i'$};
} 
\otimes 
\fineq[-0.8ex][1][0.8]{
\idline[0][0][u]
\node () at (0.1,0.1) {$j$};
\node () at (-0.2,0.2) {$j'$};
} 
\quad \text{ (unitarity) } . 
\end{equation}
The four-leg red and blue tensors represent a two-site unitary and its complex conjugate respectively, with the top/bottom legs as row/column indices $ij$/$kl$. Contraction at the bottom represents matrix multiplication, and the two $\idlineu$ tensors on the right denote the identities on the two sites. A dual unitary satisfies an additional dual unitarity relation
\begin{equation}
\label{eq:du-cond}
\fineq[-0.8ex][0.75][0.8]{
    \regbox[-0.35][0.35][blue!50];
    \idline[0][1.25];
    \idline[1][1.25][r];
    \idline[0][0];
    \idline[1][0.25][u];
    \regbox[-0.25][0.25][red!50];
    \node () at (0.15,1.5+0.1) {$i$};
    \node () at (-0.2,1.5+0.2) {$i'$};
    \node () at (0.15,-0.1) {$k$};
    \node () at (-0.3,0) {$k'$};
    \node () at (1.15,1.5+0.2) {$j$};
    \node () at (1+0.05,-0.25) {$l$};
}
= \fineq[-0.8ex][1][0.8]{
\idline[0][0][r]; 
\node () at (0.1,-0.15) {$k$};
\node () at (-0.25,-0.05) {$k'$};
\idline[0][1][u]
\node () at (0.1,1+0.1) {$i$};
\node () at (-0.2,1+0.2) {$i'$};
}
\quad \text{ (dual unitarity) }.
\end{equation}
It means that the matrix is also a unitary when viewed sideways, \ie $u_{ij,kl}u^*_{i'j,k'l} = \delta_{ii'} \delta_{kk'}$. In an index free notation, Eq.~\eqref{eq:du-cond} states that $(uF)^\Gamma $ is also unitary, where $F$ is the swap operator, and $^\Gamma$ is the partial transpose.  Examples of dual unitaries include the swap gate and quantum Fourier transform\cite{tyson_operator-schmidt_2003}.  The dynamics of the dual unitary circuit is maximally chaotic\cite{bertini_exact_2018,bertini_exact_2019-1,prosen_many_2021,bertini_entanglement_2018,gopalakrishnan_unitary_2019,piroli_exact_2020}.
In addition to being an exact example of $v_E = 1$\cite{bensa_fastest_2021,bertini_entanglement_2018,bertini_operator_2020}, its auto-correlation\cite{bertini_exact_2019-1,kos_correlations_2021} and quantum butterfly effect (the butterfly velocity) also travel at the light cone speed $1$\cite{bertini_scrambling_2020}. Its (ensemble averaged) spectrum form factor\cite{bertini_exact_2018,bertini_random_2021} exactly reproduces the random matrix behavior. Some of these properties are not exclusive though. For instance, there are circuits whose gate is not dual unitary while still having quantum butterfly velocity to be equal to $1$\cite{claeys_maximum_2020}. 
\begin{figure}[t]
\centering
\subfigure[]{
  \label{fig:circuit}	
  \includegraphics[width=0.6\columnwidth]{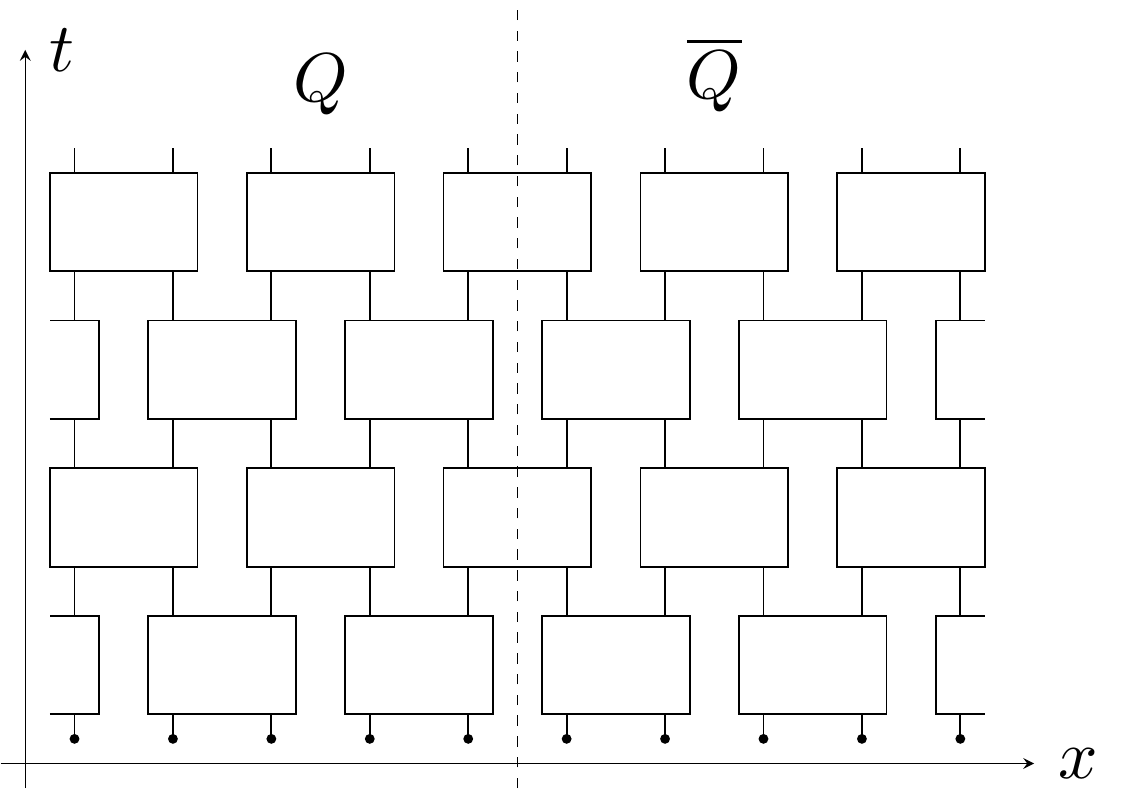}
}
\hspace{0.05\columnwidth}
\subfigure[]{
  \label{fig:4_qudits}	
  \includegraphics[width=0.25\columnwidth]{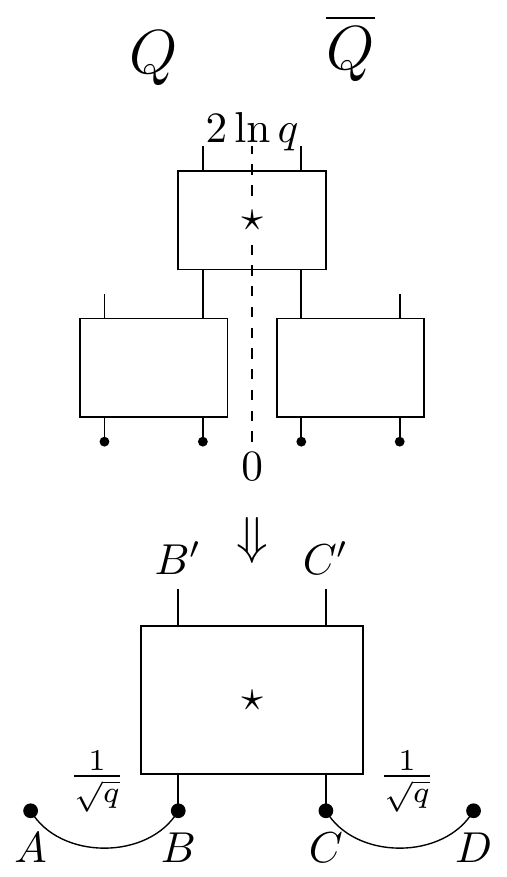}
}
\caption{(a) The brickwork circuit. Solid circles at the bottom denote product initial states. The dashed line cut the system into subsystem $Q$ and $\overline{Q}$ (b) (top) At $t = 2$, analysis of $S(Q)_{t = 2}$ reduces to 4 qudits and 3 gates. $S(Q)_{t = 2}$ is set to be the maximal $2\ln q$. (bottom) Evolution from $t = 1$ to $t = 2$ with input $\rho_{ABCD}$ and output $\rho_{AB'C'D}$.}
\label{fig:circuit_4_qubits}
\end{figure}

A natural question arises: Does $v_E = 1$\cite{FN} always imply the dual unitarity of the constituent gates? In this letter, we give a rigorous and affirmative answer, and show its robustness and generality. We prove that when $v_E = (1 - \eta)$ with a small positive $\eta$, the gate in the circuit approximately satisfies the dual unitary condition (Eq.~\eqref{eq:du-cond}) with an error up to $O(\eta^{\frac{1}{2}})$. The theorem also implies the existence of a nearby dual unitary. Our new understanding of entanglement structure in the dual circuit enables us give an alternative simpler proof that $v_E=1$ for the known exactly solvable cases. Our proofs are based on a monogamy-of-entanglement principle, namely that near-maximal entanglement between some subsystems of a state implies decoupling of other subsystems. We believe that the techniques and the entanglement structure will be widely useful in systems with rapid entanglement growth. 

{\bf Almost maximal growth by gate.} 
The entanglement velocity $v_E$ can be thought of as the long-term average rate of entanglement growth per gate.   Since the expression involves a limit, $v_E = 1$ can still be achieved if most gates have near-maximal entanglement growth. 

We start here with the $v_E \leq 1$ limit.  Follow the dashed line in Fig.~\ref{fig:circuit} from the bottom to the top. Entanglement can only change when the line pierces through a unitary gate, every other time step. The maximal growth by one gate is upper bounded by $2 \ln q$ (Lemma 1 of \cite{bennett_capacities_2003}). To prove this, let $A,B,\ldots$ be the systems in Fig.~\ref{fig:circuit}(b) and use the triangle inequality:
\begin{equation}
\label{eq:2lnq_bound}
\begin{aligned}
 S( AB') &- S(AB) \\ & =  [S(AB') - S(A) ]
 + [S(A) - S(AB) ]  \\
 &  \le S(B' ) + S(B)  \\ 
 &  \le \ln\dim(B)+\ln\dim(B') = 2 \ln(q)
\end{aligned}
\end{equation}
After $t$ time steps (assuming $t$ even), there are $t/2$ gates between $Q$ and $Q^c$, corresponding to entanglement changes $\Delta_\tau \equiv S(Q)_{2 \tau} - S(Q)_{ 2\tau - 1}$ for $\tau=1,2,\ldots,t/2$.  By \Cref{eq:2lnq_bound}, each $\Delta_\tau \leq 2\ln(q)$.  On the other hand, if $v_E=1$ then $\frac{1}{t/2}\sum_\tau \Delta_\tau \geq (1-\eta)2\ln(q)$ where $\eta\ra 0$ as $t\ra \infty$.  Thus there exists at least one $\tau$ where the entanglement increase $\Delta_\tau$ is $\geq (1-\eta)2\ln(q)$.  As we take $t\ra \infty$ this argument shows that individual gates must yield entanglement increases arbitrarily close to $2\ln(q)$.  Note that the $2 \ln(q)$ upper bound is not really used here.  We get the existence of gates with entanglement growth $\geq (1-\eta)2\ln(q)$ just because that's the average entangelment growth.  We need the upper bound only to interpret this as near-maximal.

{\bf A 4-qudit model.} 
A simple version of the relation between dual unitarity and maximal entanglement growth can be seen in a 4-qudit example. Suppose we have $S(Q)_{t=0}=0$ and $S(Q)_{t= 2} = 2\ln q$ in Fig.~\ref{fig:circuit}, then for the sake of entanglement $S(Q)$, we only need to consider 4 qudits and 3 gates(Fig.~\ref{fig:4_qudits}). (Later we will generalize this to the case where the initial entanglement may be large.) We label the four qudits at the slice of $t = 1$ as $A$, $B$, $C$ and $D$. The gate evolves $B$ and $C$ to $B'$ and $C'$ (Fig.~\ref{fig:4_qudits} bottom). Our assumption of maximal entanglement growth means that $S(AB') - S(AB) = 2 \ln q$.

The input and output states can be determined from the entropies. We have $S(AB) = 0$, due to the product initial state and absence of gate across $AB$ and $CD$ at $t = 1$. Thus $AB'$ is maximally mixed. By tracing out $B'$, so is $A$ (for $t = 2$ and consequently for $t = 1$). $B$ therefore forms a Bell state with $A$ at $t = 1$. Similarly $C$ forms a Bell state with $D$. We denote the Bell state as a curved line connecting the qudits in Fig.~\ref{fig:4_qudits} bottom. 

In a graphical notation similar to Eq.~\ref{eq:u-cond} and Eq.~\ref{eq:du-cond}, we rewrite $\rho_{AB'}$ in two ways
\begin{equation}
\label{eq:4_qudit_rho_ABp}
 \frac{1}{q^2} \uucontract[][tr][epr][eprtr]  = \frac{1}{q^2} \idlineu\otimes \idlineu.
\end{equation}
On the LHS, the input state -- two separate Bell pairs -- is conjugated by $u$ (red) and $u^\dagger$ (blue). Partial trace at $C'$, $D$ denoted by the closed loop gives $\rho_{AB'}$. The open $\idlineu$-shape symbol denotes the maximally mixed states at $A$ and $B'$. Canceling the normalization factor $\frac{1}{q^2}$, Eq.~\eqref{eq:4_qudit_rho_ABp} is an alternative way to write down the dual unitary condition in Eq.~\eqref{eq:du-cond}. Thus we see that maximal entanglement growth implies dual unitarity in our example.


{\bf Approximate maximal entangling.} 
We extend the intuition in the 4-qudit toy model to the case where entanglement growth is almost maximal.  This could arise if $v_E=1$ and individual gates approach but do not necessarily achieve this limit; alternately we might have $v_E$ close to, but not equal to, 1.
In Theorem~\ref{th:close_to_dual}) we will  derive entropy bounds to analyze the input and output states, yielding an approximate dual unitary condition.

More formally, let us consider at time slice $t$ there is a gate on the dashed line Fig.~\ref{fig:circuit} which is nearly maximally entangling. The gate $u$ acts on qudits $B$ and $C$, while $A$ ($D$) now denotes the collection of qudits to the left (right) of $B$ ($C$), see Fig.~\ref{fig:sie}. A similar setup has been employed in the study of entanglement generation by Hamiltonians, where $A$ and $D$ are the auxiliary qudits\cite{bravyi_upper_2007,lieb_upper_2013,van_acoleyen_entanglement_2013,marien_entanglement_2016,avery_universal_2014}. Unitary gates acting exclusively on $A$ or $D$ do not change $S(AB')$, so are ignored.

\begin{figure}[h]
\centering
\subfigure[]{
  \label{fig:sie}	
  \includegraphics[width=0.35\columnwidth]{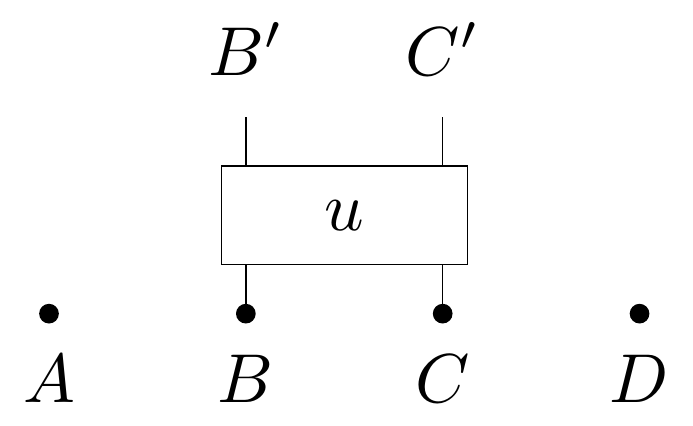}
}
\hfill
\subfigure[]{
  \label{fig:ABCD_decomp}	
  \includegraphics[width=0.55\columnwidth]{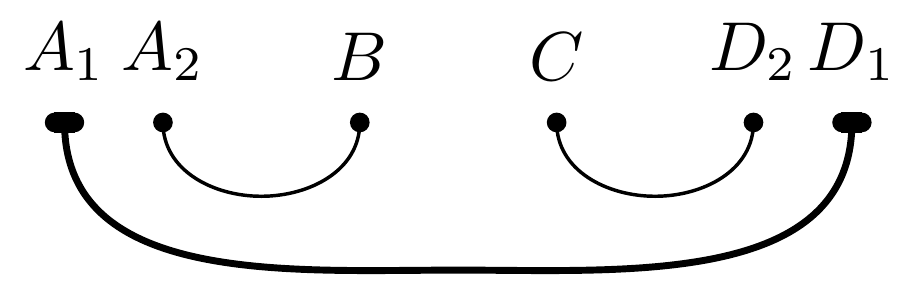}
}
\caption{(a) The 4-party setup with $2\ln q - \epsilon$ entanglement generation by a unitary $u$ acting on qudits $B$,$C$. $A$, $D$ are auxiliary system with arbitrary (finite) dimensions. (b) The distillable entanglement structure. $A$ ($D$) is partitioned into $A_1$ ($D_1$) and qudit $A_2$ ($D_2$). $A_2 B$ and $D_2 C$ are Bell pairs. }
\label{fig:sie_decomp}
\end{figure}

\begin{theorem}[proximity to dual unitarity]
\label{th:close_to_dual}
Let $u$ act as in Fig.~\ref{fig:sie} such that
\begin{equation}
\label{eq:S_ABp_S_AB}
S(AB') - S(AB) = 2 \ln q - \epsilon
\end{equation}
then
\begin{equation}
\label{eq:close_to_daul}
\norm{\frac{1}{q}\idlineu \otimes \frac{1}{q} \idlineu - \frac{1}{q^2}\uucontract[][tr][epr][tr] }_1 \le O( \epsilon^{\frac{1}{2}} )
\end{equation}
\end{theorem}

When $v_E = 1$, $\epsilon$ goes to zero for a sequence of gates along the dashed line in Fig.~\ref{fig:circuit}. Hence $v_E = 1$ implies that the gate (anywhere) is dual unitary. If $v_E = 1 -
\eta$, then the entanglement growth for a sequence of gates along the dashed line can converge to $2 \ln q ( 1- \eta
)$. Theorem~\ref{th:close_to_dual} indicates that the dual unitary condition is satisfied up to error of order $\eta^{\frac{1}{2}}$.

When Eq.~\eqref{eq:close_to_daul} holds, it there a nearby dual unitary? When $q=2$, the explicit parameterization of the set of dual unitaries gives us explicit bounds on the distance to the nearest dual unitary; for $q>2$, we know only non-explicit bounds.
\begin{theorem}\label{thm:nearby-DU}
If $u$ has $v_E = 1 - \eta$ for $0 < \eta < 1$, then $\exists$ a dual unitary $u_{\times}$, s.t. it is close to the gate $u$ up to an error
\begin{equation}
  \norm{u - u_{\times} }_1 \le
  \begin{cases}
    O( \eta^{\frac{1}{4}}) & \text{if $q=2$} \\
    f_q(\eta) & \text{if $q>2$},
    \end{cases}
  \end{equation}
where $f_q(\eta)\rightarrow 0$ as $\eta\rightarrow 0$.
\end{theorem}

The rest of this section gives a proof sketch of \Cref{th:close_to_dual}. \Cref{thm:nearby-DU} is proved in \appref{app:closeness_q_2}. 

Following the intuition of the 4-qudit example, we use entropy bounds to deduce the entanglement structure. First, we show that near-maximal entanglement increases require that $B$ and $C$ be nearly maximally entangled with $A$ and $D$, respectively.
\begin{lemma}
\label{le:zigzig}
Let $u$ act as in Fig.~\ref{fig:sie} and assume entanglement growth in Eq.~\eqref{eq:S_ABp_S_AB}. Then
\begin{align}
\label{eq:lemma_zigzag}
-S(B|A) = S(A) - S(AB)& \ge \ln q - \epsilon\\
-S(D|C) = S(D) - S( CD)& \ge \ln q - \epsilon. 
\end{align}
\end{lemma}
The lemma can be proved by telescoping $S(A) - S(AB') + S(AB') - S(AB)$ and using sub-additivity (see \appref{app:char_max_ent} for details). 

The subsystem $AB$ contains one extra qudit $B$ than $A$, yet its entanglement is at least $\ln q - \epsilon$ smaller. This almost maximal difference implies that $-S(B|A) \geq \ln(q)-\eps$ entanglement can be asymptotically distilled from many copies of the state~\cite[Chaps 11, 24]{Wilde-book}.  In fact, near-maximal entanglement can be distilled even from a single copy of the state using local operations, using a decoupling argument as in Ref.~\cite{abeyesinghe_mother_2009}.
\begin{lemma}[near-maximal distillable entanglement in input]
\label{le:two-epr-pair}
  Up to unitary transformations exclusively in $A$ or $D$, the input state $\rho_{\scriptstyle ABCD}$ can be approximated by
\begin{equation}
\label{eq:sigma_ABCD}
\sigma_{ABCD} = |\alpha_{A_2 B} \rangle \langle \alpha_{A_2 B} | \otimes  \sigma_{A_1 D_1}  \otimes | \beta_{CD_2 }\rangle  \langle  \beta_{C D_2 } |
\end{equation}
s.t.
\begin{equation}
\label{eq:rho_ABCD}
\begin{aligned}
&\norm{\rho_{ABCD}  -  \sigma_{ABCD}}_1\le O(\epsilon^{\frac{1}{2}}).
\end{aligned}
\end{equation}
Here $A_1, A_2$ ($D_1, D_2)$ are partitions of $A$ ($D$), and $A_2$ ($D_2$) is a qudit. $|\alpha_{A_2 B} \rangle $ and $|\beta_{CD_2} \rangle $ are maximally entangled, i.e. $S(A_2)_\alpha = S(D_2)_\beta =  \ln q $. 
\end{lemma}
Fig.~\ref{fig:ABCD_decomp} depicts the structure of $\sigma_{ABCD}$ in the theorem. Using similar notation as in Eq.~\eqref{eq:4_qudit_rho_ABp} for the Bell state, \Cref{eq:rho_ABCD} can be written as 
\begin{equation}
\label{eq:rho_ABCD_fig}
\norm{\rho_{ABCD} - \sigma_{A_1 D_1} \otimes \idlineepr_{A_2B} / q 
\otimes \idlineepr_{CD_2} / q 
}_1 \le O( \epsilon^{\frac{1}{2}})
\end{equation}

The main idea is to recover $\rho_{ABCD}$ from an approximation of $\rho_{BCD}$. Using \Cref{le:zigzig}, we can derive entropy bounds (specifically $I( B; CD ) \le \epsilon$ and $S_B \ge \ln q - \epsilon$) which lead to approximate decoupling $\norm{\rho_{BCD} - \idlineu / q \otimes \rho_{CD}}_1 \le O( \epsilon^{\frac{1}{2}} )$. Then we purify to find the partition $A_1 A_2$, so that $B$ is almost maximally entangled with $A_2$ (note this $A_2$ was previously $A$ in the 4-qudit model).  We repeat the argument for $D_1,D_2$, taking care not to break the $A_2B$ entanglement in the process; this is described in full detail in \appref{subapp:two_epr_pair_proof}.


 
Now we discuss the constraints that apply to the output state $\rho_{AB'}$, which should be nearly maximally mixed on $A_2B'$ after tracing out $A_1$. 
\begin{lemma}[almost maximally mixed output]
\label{le:rho_ABp}
Given the configurations in Fig.~\ref{fig:sie} and entanglement growth in Eq.~\eqref{eq:S_ABp_S_AB}, we have 
\begin{equation}
\label{eq:rho_ABp_decomp}
\norm{\rho_{A_1A_2B'} - \sigma_{A_1} \otimes  \idlineu_{A_2} /q  \otimes  \idlineu_{B'} /q }_1 \le O( \epsilon^{\frac{1}{2}} )
\end{equation}
where $\idlineu/q$ denotes a maximally mixed state, $\sigma_{A_1}$ is the reduced state from $\sigma_{ABCD}$ in Eq.~\eqref{eq:sigma_ABCD}. 
\end{lemma}
Again, we use the entropy bounds (specifically $I(A;B') \le \epsilon$ and $S(B') \le \ln q - \epsilon$) to deduce to deduce approximate decoupling $\norm{ \rho_{AB'} - \rho_{A} \otimes \idlineu /q}_1 \le O( \epsilon^{\frac{1}{2}} )$
We then replace the $\rho_A$ by the approximate $\sigma_{A_1} \otimes  \idlineu_{A_2} /q $ from the known structure in \Cref{le:rho_ABp}. 

With these ingredients, we can prove Theorem~\ref{th:close_to_dual}.
\begin{proof}
Taking partial trace of $C'D$ in Eq.~\eqref{eq:rho_ABCD_fig}, we have the approximation from the input side, 
\begin{equation}
\norm{\rho_{AB'} - \sigma_{A_1} \otimes \frac{1}{q^2}\uucontract[][tr][epr][tr] }_1 \le O( \epsilon^{\frac{1}{2} })
\end{equation}
By the output side, we replace $\rho_{AB'}$ with the maximal mixed state in Lemma~\ref{le:rho_ABp}, 
\begin{equation}
\norm{\sigma_{A_1}\otimes \frac{1}{q} \idlineu   \otimes \frac{1}{q} \idlineu   - \sigma_{A_1} \otimes \frac{1}{q^2}\uucontract[][tr][epr][tr] }_1 \le O( \epsilon^{\frac{1}{2} })
\end{equation}
Taking partial trace in $A_1$ does not increase the distance. We thus obtain \Cref{eq:close_to_daul}. 
\end{proof}

{\bf Mechanism for $v_E = 1$.} 
We have shown that dual unitarity is a necessary condition for $v_E = 1$. For sufficiency, it has been proven that $v_E = 1$ for dual unitary circuits even at finite times, given special classes of initial states. These initial states are the ``separating states''\cite{bertini_entanglement_2018} for the self-dual kicked Ising model and the ``solvable'' matrix product state (MPS) ansatz\cite{piroli_exact_2020} that preserves the unitarity at the boundary after rotating the circuit. 


\begin{figure}[h]
\centering
\includegraphics[width=0.9\columnwidth]{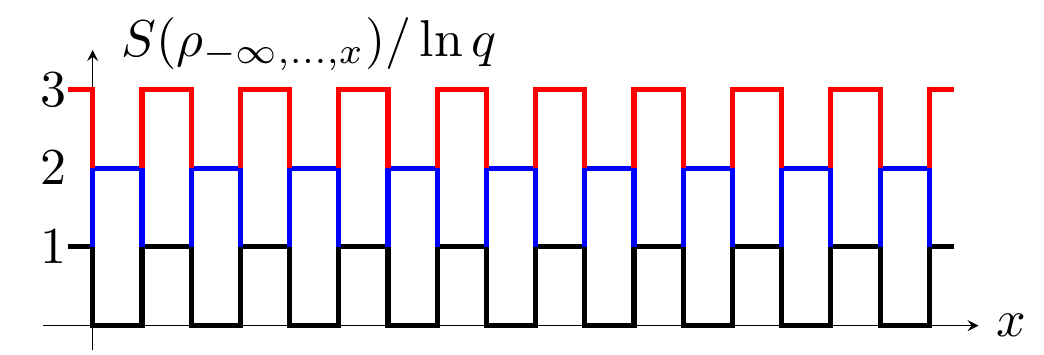}
\caption{Entanglement entropy of sites from $-\infty$ to $x$ in the unit of $\ln q$ for $t = 1$ (black), $t = 2$ (blue) and $t = 3$ (red). $v_E = 1$ at finite time. At $t = 1$, the entanglement alternates between $0$ and $\ln q$. Applying dual unitary gates at the valleys relay this zigzag pattern. }
\label{fig:zigzag}
\end{figure}
We reproduce these results in a simpler way through the entanglement structure developed above. 
\begin{theorem}[dual unitarity relays the zigzag entanglement pattern]
  Suppose at the $t = 1$ time slice, the entanglement across bonds alternates between $\ln q$ or $0$, and we apply a dual unitary circuit.  Then at all even time steps $t$ we have
  \begin{equation}
    S(Q)_{t}=t \ln q.
  \end{equation}
  This is true even when the dual unitary gates are not identical across the circuit. 
\end{theorem}
\begin{proof}
  At $t = 1$, the entanglement profile is given by the black curve in Fig.~\ref{fig:zigzag}. There are peaks whose value is $\ln q$ and valleys whose value is $0$. Since the valley has $\ln q$ entanglement smaller than its neighbors (\Cref{le:zigzig} with $\epsilon = 0$), the input state locally has the exact distillable entanglement structure in \Cref{fig:ABCD_decomp}(\Cref{le:two-epr-pair} with $\epsilon = 0$). When a dual unitary gate acts at the valley, dual unitarity guarantees to increase the entanglement by $2 \ln q$ (the 4-qudit model). A valley becomes a peak. In brickwork circuit structure, the gate always acts on valleys. Thus the circuit interchanges the role of peak and valley in one step, yet still maintain the entanglement difference to be $\ln q$. For example, the red and blue lines in Fig.~\ref{fig:zigzag} depicts the entanglement profile at $t = 2$ and $t = 3$. We see that the dual unitary gates can relay the zigzag entanglement pattern, while always generates $2 \ln q$ entanglement in each step, even if they are different across the circuit. Hence the exact relation $S(Q)_t = t \ln q$ at even steps. 
\end{proof}

We find that both the separating states\cite{bertini_entanglement_2018} (after $1$ or $2$ steps) and solvable MPS \cite{piroli_exact_2020} can initiate a zigzag pattern (blue line in Fig.~\ref{fig:zigzag}) after $1$ to $2$ steps (see \appref{app:sep_states_in_skdi} and \appref{app:smps}). Thus its entanglement growth can achieve $v_E = 1$ without the need for an asymptotic limit. Starting from a general short-range entangled initial states, we conjecture that the zigzag pattern can be dynamically generated if the circuit is not integrable. Then an almost maximal entanglement growth could be sustained, achieving $v_E = 1$ as $t\ra\infty$.

{\bf Conclusions.--} One two-site unitary gate can at most increase the entanglement by $2\ln q$. The maximal entanglement velocity condition forces the gates along the entanglement cut to achieve this $2\ln q$ limit. Our analysis reveals that such (approximate) $2\ln q$ growth is always accompanied by a distillable entanglement structure locally for the input state (\Cref{le:two-epr-pair}), and [approximate] dual unitarity condition for the acting gate (\Cref{th:close_to_dual}). We use a one-dimensional brickwork circuit, but the results can be straightforwardly extended into higher dimensions and other circuit architectures. A more challenging extension is the continuum limit with Lorentz invariance~\cite{casini_spread_2015}.

It is known that in a random unitary (brickwork) circuit, $v_E $ is $ 1 - \mathcal{O}( \frac{1}{\ln(q)})$ at large $q$\cite{jonay_coarse-grained_2018,zhou_emergent_2019}. Thus each gate increases entanglement by $2 \ln q - \epsilon$, with $\epsilon$ an $\mathcal{O}(1)$ number in $q$. By \Cref{th:close_to_dual}, we infer that a random unitary has an $\mathcal{O}(1)$ distance to a dual unitary. This is in fact consistent with an average fidelity calculation, yielding $\frac{8}{3\pi} + \mathcal{O}(q^{-2})$ (see \appref{app:ru_fidelity}).  We can compare with pure-state entanglement using the Choi-Jamio\l{}kowski isomorphism. In this way, unitary operators correspond to maximally entangled states.   A random pure state on two qudits has entanglement $\ln q - \mathcal{O}(1)$, and has an expected fidelity with a maximally entangled state of $\frac{8}{3\pi} + \mathcal{O}(q^{-2})$ (see \appref{app:ru_fidelity}). Thus in both cases, we  see an $\cO(1)$ deviation from maximal entanglement, whether measured in entropy or fidelity.

One route to $v_E=1$ with dual unitary gates is the ``zigzag pattern'' of entanglement we describe above. Locally, the zigzag pattern implies a distillable entanglement structure, but can we characterize the  global structure of those states? Easy solutions lie at the integrable points where entanglement is carried by free streaming quasi-particles\cite{calabrese_quantum_2007}, e.g.~a swap circuit (note: swap is a dual unitary) with a dimer initial state.  In most cases, dual unitary circuits are chaotic, and we have no good description of the global structure of the corresponding states with zig-zag entanglement.

There are situations when the zigzag pattern is absent in the initial state or early time evolution. This is the case for a self-dual kicked Ising model with random product initial state~\cite{bertini_entanglement_2018}. An open question is whether the zigzag pattern and locally distillable entanglement structure can be dynamically generated, and maintained when projective measurement is present\cite{skinner_measurement-induced_2019,li_quantum_2018,chan_unitary-projective_2019,choi_quantum_2020,gullans_dynamical_2020-1}. 

Finally, the almost maximally mixed output state ($A B'$ or $C' D$) provides an economical way to classically simulate the reduced density matrix in terms of the tensor network algorithm. In a circuit with gates close to dual unitary (perhaps even random circuits with large $q$), significant speedup for observables confined in half of the system is anticipated. We leave this to future work.



\acknowledgements
Both authors were funded by NTT Research Award AGMT DTD 9.24.20.
AWH is also supported by NSF grants CCF-1729369 (STAQ), PHY-1818914 (EPiQC) and OMA-2016245 (QLCI CIQC).

We thank Yichen Huang for helping check the proof of Theorem 1, and Daniel Ranard for suggesting the continuity argument for $q > 2$ in Theorem 2. 
We thank Pavel Kos, Austen Lamacraft, Pieter Claeys, Wen Wei Ho and Bruno Bertini for commenting the manuscript.
We thank Wen Wei Ho and Bruno Bertini for explaining the properties of the solvable MPS ansatz.

\appendix

\section{Characterization of almost maximal entanglement growth}
\label{app:char_max_ent}
In the main text, we argue that there is a sequence of gates in the circuit that generates entanglement converging to $2 \ln q$. Our strategy to prove the proximity to the dual unitary condition is to examine one such gate with entanglement growth $2 \ln q - \epsilon$ for a small $\epsilon$. 

We partition the time slice in which gate resides to have 4 parties labeled as $A$, $B$, $C$ and $D$. The gate acts on qudit $B$ and $C$, while $A$ ($D$) are the collection of qudits to the left (right) of the gate. The demonstration is reproduced in Fig.~\ref{fig:sie_setup_app}. 

\begin{figure}[h]
\centering
\includegraphics[width=0.4\columnwidth]{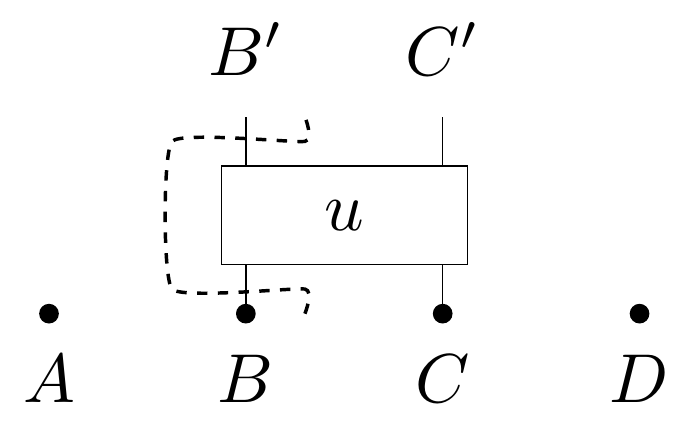}
\caption{The 4-party setup with $2\ln q - \epsilon$ entanglement generation by the 2-site unitary $u$. $B$,$C$ are qudits, $A$, $D$ are auxiliary system with arbitrary (finite) dimensions.}
\label{fig:sie_setup_app}
\end{figure}

The 4-party setup has more general applications in the research of entanglement generation rate problem\cite{marien_entanglement_2016,avery_universal_2014,bravyi_upper_2007,lieb_upper_2013,van_acoleyen_entanglement_2013}, in which the unitary evolution is governed by a continuous Hamiltonian evolution. In the setup, $B$ and $C$ are qudits, $A$ and $D$ are ancilla systems with unspecified dimensions. The ancillas help to achieve a larger growth bound. For instance, without the ancillas, the maximal entanglement growth by the gate is $\ln q$; with the ancillas, the limit $2 \ln q$ is possible. 

For our purpose, we have $S( AB') - S(AB) = 2 \ln q - \epsilon$. Here we prove important properties stemmed from this condition, and supplement a few deferred proofs in the main text. 

\subsection{Entropy bounds of the input and output states}

In this subsection, we derive various entropy bounds of the input and output states in the configurations of Fig.~\ref{fig:sie_setup_app} as a consequence of almost maximal entanglement growth
\begin{equation}
\label{eq:app_max_ent_cond}
S( AB') - S(AB) = 2 \ln q - \epsilon. 
\end{equation}
We will repeatedly use the bounds of the conditional entropies in Lemma 1, which we now prove
\begin{align}
S(A) - S(AB) \ge \ln q - \epsilon  \label{eq:_SA_m_SAB}\\
S(D) - S(CD) \ge \ln q - \epsilon 
\end{align}
\begin{proof}
We only show Eq.~\eqref{eq:_SA_m_SAB}. We use $S(AB')$ in a telescoping sum 
\begin{equation}
\begin{aligned}
S(A) - S(AB) &= S(A) - S(AB' ) + S(AB') - S(AB) \\
&\ge -S(B') + 2 \ln q - \epsilon  \ge \ln q - \epsilon.
\end{aligned}
\end{equation}
where we use sub-additivity in the first inequality and the entanglement bound in the second. 
\end{proof}

First of all, on the input side $\rho_B$, $\rho_C$ are almost maximally mixed, and so are $\rho_{B'}$, $\rho_{C'}$ on the output side. 
\label{subapp:in_out_states}
\begin{claim}
\begin{align}
S(B) \ge \ln q - \epsilon, \quad S(B') \ge \ln q - \epsilon \label{eq:S_B_max}\\
S(C) \ge \ln q - \epsilon, \quad S(C') \ge \ln q - \epsilon 
\end{align}
\end{claim}
\begin{proof}
Without loss of generality, we only prove Eqs.~\eqref{eq:S_B_max}. We insert $S(A)$ into the Eq.~\eqref{eq:app_max_ent_cond}, leading to
\begin{equation}
[ S(AB') - S(A) ] + [ S(A) - S(AB)]   \ge 2 \ln q - \epsilon 
\end{equation}
On the LHS, we use triangle inequality and Araki-Lieb inequality for the terms in the two brackets respectively: 
\begin{align}
S(AB') - S(A) \le S(B') \\
S( A) - S(AB) \le S(B). 
\end{align}
Hence
\begin{equation}
S(B) + S(B') \ge 2 \ln q - \epsilon 
\end{equation}
Since both $S(B)$ and $S(B')$ are upper bounded by $\ln q$, we obtain Eqs.~\eqref{eq:S_B_max}. 
\end{proof}

In fact, the two-qudit state $\rho_{BC}$ is almost maximally mixed. 
\begin{claim}
\begin{equation}
\label{eq:S_BC_bound}
S(BC) \ge 2 \ln q - 2\epsilon 
\end{equation}
\end{claim}
\begin{proof}
We use $S( \rho_{BC} || \rho_B \otimes \frac{\I_C}{q} )$ as a proxy. On one hand, we have using Eq.~\eqref{eq:S_B_max} 
\begin{equation}
\begin{aligned}
  S( \rho_{BC} || \rho_B \otimes \frac{\I_C}{q} ) &= S(B) + \ln q  - S( BC) \\
  &\ge 2 \ln q - \epsilon - S(BC) .
\end{aligned}
\end{equation}
On the other hand, by the data processing inequality and the bound for the conditional entropy, we have
\begin{equation}
\begin{aligned}
&S( \rho_{BC} || \rho_B \otimes \frac{\I_C}{q} ) \le S( \rho_{ABC} ||\rho_{AB} \otimes \frac{\I_C}{q} )\\
=\,& S_{AB} + \ln q - S_{ABC}  = \ln q + S_{CD} - S_D  \le \epsilon. 
\end{aligned}
\end{equation}
\end{proof}

Then we show that on the input side, $AB$ is almost product with $C$, and $CD$ is almost entangled with $B$. These are embodied in the almost vanishing mutual information. 
\begin{claim}
\begin{align}
I(AB;C) \le \epsilon \label{eq:I_AB_C_e}\\
I(B; CD) \le \epsilon   \label{eq:I_B_CD_e}
\end{align}
\end{claim}
\begin{proof}
Without loss of generality, we only prove Eq.~\eqref{eq:I_B_CD_e}. It is straightforward write down $I(B; CD)$ in terms of the entropies in $AB$ for a pure state $\rho_{ABCD}$: 
\begin{equation}
\begin{aligned}
I(B; CD) &= S(B) + S(CD) - S(BCD) \\
&= S(B) + S(AB) - S(A) 
\end{aligned}
\end{equation}
Then we apply the bound of conditional entropy in Eq.~\eqref{eq:_SA_m_SAB}, and recognize $S(B) \le \ln q$
\begin{equation}
I(B; CD) \le S(B) - \ln q + \epsilon  \le \epsilon
\end{equation}
\end{proof}

We note that on the output side, $A$ and $B'$ are nearly decoupled, as are $D$ and $C'$.  Again this can be quantified via mutual information.  
\begin{claim}
\begin{align}
I(A;B') \le \epsilon \label{eq:I_A_Bp}\\
I(C'; D) \le \epsilon 
\end{align}
\end{claim}
\begin{proof}
  Without loss of generality, we prove Eq.~\eqref{eq:I_A_Bp}. In this case we insert $S(AB)$ in the expression of $I(A;B')$
\begin{equation}
\begin{aligned}
I(A;B') = &S(A) + S(B') - S(AB') \\
= [S(A) - S(AB)]& + [S(B')] - [S(AB') - S(AB)]\\
\end{aligned}
\end{equation}
In the three brackets, we respectively apply bound on the conditional entropy in Eq.~\eqref{eq:_SA_m_SAB}, entropy upper bound $S(B') \le \ln q $ and almost maximal entanglement growth assumption in Eq.~\eqref{eq:app_max_ent_cond}. We then have 
\begin{equation}
I(A;B') \le  \ln q + \ln q - 2\ln q + \epsilon = \epsilon 
\end{equation}
\end{proof}

\subsection{Factorization of the input and output states}

Given the entropy bounds above, we use the following theorem to factorize the state: 
\begin{theorem}
Let $A$, $B$ have dimension $d_A$ and $d_B$. If $S(A) \ge \ln d_A - \epsilon$, and $I(A;B) \le \epsilon$, then 
\begin{equation}
F( \rho_{AB}, \frac{I_A}{d_A} \otimes \rho_B ) \ge \exp( - \epsilon ) 
\end{equation}
where $\I_A$ is the identity operator on $A$. It also implies a bound on trace distance:
\begin{equation}
\label{eq:rho_AB_I_A_dA_rho_B}
\norm{\rho_{AB} - \frac{\I_A}{d_A} \otimes \rho_B }_1 \le O( \epsilon^{\frac{1}{2}} ) 
\end{equation}
\end{theorem}

The fidelity measure has two advantages over the trace distance here: the fidelity does not increase for the optimal purification required later; and the bound also works even when $\epsilon$ is not small (see App.~\ref{app:dist_mea}).
\begin{proof}
We consider the relative entropy
\begin{equation}
\begin{aligned}
S( \rho_{AB } || \frac{\I_A}{d_A} \otimes \rho_B) &= 
\ln d_A + S(B) - S(AB) \\
&= I(A; B) + \ln d_A - S_A \le 2 \epsilon . 
\end{aligned}
\end{equation}
According to Theorem~\ref{th:rel_S_F} (monotonicity of sandwiched R\'enyi divergence), 
\begin{equation}
\begin{aligned}
F( \rho_{AB } || \frac{I_A}{d_A} \otimes \rho_B ) &\ge \exp( - \frac{1}{2} S( \rho_{AB } || \frac{I_A}{d_A}  \otimes \rho_B )  ) \\
&\ge \exp( - \epsilon ) .
\end{aligned}
\end{equation}
Using the relation between the trace distance and fidelity, we also have Eq.~\eqref{eq:rho_AB_I_A_dA_rho_B}. 
\end{proof}

\begin{corollary} 
\label{coro:ABp_BCD}
By Eq.~\eqref{eq:I_A_Bp} and Eq.~\eqref{eq:I_B_CD_e}, we have the approximate factorization of $\rho_{AB'}$ and $\rho_{BCD}$
\begin{align}
F( \rho_{AB'} , \rho_A \otimes \frac{\I_{B'}}{q} ) &\ge \exp( -\epsilon ) \label{eq:rho_ABp_rhoA_I}\\
F( \rho_{BCD}, \frac{\I_B}{q} \otimes \rho_{CD}) &\ge \exp( - \epsilon )  \label{eq:rho_BCD_I_B_rho_CD}
\end{align}
The corresponding trace distance is bounded by $O( \epsilon^{\frac{1}{2}})$. Similarly we have the approximation on the side: 
\begin{align}
F( \rho_{C'D} ,  \frac{\I_{C'} }{q} \otimes \rho_D) &\ge \exp( - \epsilon ) \\
F( \rho_{ABC} , \rho_{AB} \otimes \frac{\I_C}{q})  &\ge \exp( - \epsilon ). \label{eq:rho_ABC_rho_AB_IC}
\end{align}
\end{corollary}

\subsection{The near-maximal distillable entanglement structure of the input state}
\label{subapp:two_epr_pair_proof}

We uncover the structure of the input state $\rho_{ABCD}$ from the closeness (in fidelity) of $\rho_{BCD}$ to a tensor product of a maximally mixed state on $B$ and the state $\rho_{CD}$(Eq.~\eqref{eq:rho_BCD_I_B_rho_CD}). 

To simplify the notation, in this section we use $\psi \equiv \dyad{\psi}$ to represent the density matrix of the pure state $|\psi \rangle $. 
\begin{lemma}
Suppose
\begin{equation}
F( \rho_{BCD}, \frac{\I_B}{q} \otimes \rho_{CD}) \ge \exp( - \epsilon ).
\end{equation}
Fix any $|\phi_L \rangle = |\alpha  \rangle_{A_2 B}  |\mu \rangle_{A_1 CD}$ satisfying
\begin{align}
\tr_{A_1} ( \mu_{A_1 CD}  ) &= \rho_{CD}    \\
\tr_{A_2} ( \alpha_{A_2 B}  ) &= \frac{\I_B}{ q }.
\end{align}
with $A_1, A_2$ a partition of $A$.

Then there exists a unitary $U_A$ acting only on $A$ s.t.
\begin{equation}
\label{eq:uhlmann_rho_BCD}
\begin{aligned}
F( U_A \rho_{ABCD} U_{A}^\dagger, \phi_L  ) &= 
  F( \rho_{BCD}, \frac{\I_B}{q} \otimes \rho_{CD})\\
  &\ge \exp( - \epsilon ). 
\end{aligned}
\end{equation}
and
\begin{equation}
\label{eq:app_phi_L}
  \norm{U_A \rho_{ABCD} U_{A}^\dagger-\phi_L }_1 \le O( \epsilon^{\frac{1}{2}} ) 
\end{equation}
\end{lemma}

This Lemma is an immediate consequence of Uhlmann's theorem.
The state $|\phi_L \rangle $ comes from purifying $\I_B/q$ and $\rho_{CD}$ into $A_1$ and $A_2$ respectively.  This gives us the distillable entanglement structure on the $AB$ side of the state.

We can show the distillable entanglement structure on the $CD$ side of the state from the corresponding condition for $\rho_{ABC}$ in Eq.~\eqref{eq:rho_ABC_rho_AB_IC}. The target purification we seek to approximate is $|\tilde{\phi}_R \rangle =   \ket{\tilde{\nu}}_{AB D_1}   \ket{\beta}_{C D_2 } $. Here $D_1, D_2$ is a partition of $D$, and
\begin{align}
\tr_{D_1} ( \tilde{\nu}_{AB D_1}  ) &= \rho_{AB}    \\
\tr_{D_2} ( \beta_{D_2 C}  ) &= \frac{\I_C}{ q}.
\end{align}
Again, Uhlmann's theorem guarantees the existence of unitary $U_D$ acting only on $D$, s.t. 
\begin{equation}
F( U_D \rho_{ABCD}  U_{D}^\dagger, \tilde{\phi}_R   ) \ge \exp( - \epsilon ). 
\end{equation}
For later convenience, we define
\begin{equation}
\begin{aligned}
|\phi_R \rangle  &\equiv U_A |\tilde{\phi}_R \rangle  =  |\nu \rangle_{ABD_1}   | \beta  \rangle_{C D_2} \\
|\nu \rangle_{ABD_1}   &\equiv U_A |\tilde{\nu}_{ABD_1} \rangle,   \\
\end{aligned}
\end{equation}
then
\begin{equation}
\begin{aligned}
F( U_D U_A\rho_{ABCD} U_A^\dagger  U_{D}^\dagger, \phi_R   )
 = F( U_D \rho_{ABCD}  U_{D}^\dagger, \tilde{\phi}_R   ).
\end{aligned}
\end{equation}
Therefore, for the state $U_A \rho_{ABCD} U_A^\dagger$ in \Cref{eq:app_phi_L}, $\exists$ $U_D$ acting only on $D$, s.t.
\begin{equation}
F( U_D U_A \rho_{ABCD} U_A^\dagger  U_{D}^\dagger, \phi_R  ) \ge \exp( - \epsilon )
\end{equation}
and
\begin{equation}
\label{eq:app_phi_R}
\norm{U_D U_A \rho_{ABCD} U_{A}^\dagger U_{D}^\dagger- \phi_R }_1 \le O( \epsilon^{\frac{1}{2}} ). 
\end{equation}

Now we prove that the near-maximal distillable entanglement structure can simultaneously hold on the two sides, that is up to unitary transformation on $A$ and $D$, $\rho_{ABCD}$ is close to a purification
\begin{equation}
  \sigma_{ABCD} = \dyad{\alpha}_{A_2 B} \otimes  \sigma_{A_1 D_1}  \otimes \dyad{\beta}_{C D_2 }.
\end{equation}

\begin{proof}[Proof of \Cref{le:two-epr-pair}]
From Eq.~\eqref{eq:app_phi_L} and Eq.~\eqref{eq:app_phi_R}, we can relate ${\phi_L}$ and ${\phi_R}$
\begin{equation}
\label{eq:phi_L_phi_R}
\norm{ U_D {\phi_L}  U_D^\dagger -  {\phi_R}}_1 \le O(\epsilon^{\frac{1}{2}}).
\end{equation}
We take partial trace on $A_2 B$ in Eq.~\eqref{eq:phi_L_phi_R} and get (recall that ${\phi_L} = {\alpha_{A_2 B}}\ot {\mu_{A_1 CD }}$, ${\phi_R} = {\nu_{A B D_1 }}\ot {\beta_{C D_2}} $)
\begin{equation}
\begin{aligned}
\norm{U_D {\mu_{A_1 CD }} U_D^\dagger  -  \sigma_{A_1 D_1}   \otimes {\beta_{C D_2}} }_1 \le  O(\epsilon^{\frac{1}{2}})
\end{aligned}
\end{equation}
where
\begin{equation}
\sigma_{A_1 D_1 } = \tr_{A_2 B}(  {\nu_{A B D_1 }} ) .
\end{equation}
Taking tensor product with ${\alpha_{A_2 B}}$, the distance remain unchanged, 
\begin{equation}
\norm{U_D {\phi_L}  U_D^\dagger  - {\alpha_{A_2 B}} \otimes  \sigma_{A_1 D_1}   \otimes {\beta_{C D_2}} }_1 \le O(\epsilon^{\frac{1}{2}}). 
\end{equation}
Finally, we replace $U_D {\phi_L} U_D^\dagger$ by $U_D U_A\rho_{ABCD} U_A^\dagger   U_D^\dagger$ by Eq.~\eqref{eq:app_phi_L} and get
\begin{equation}
\label{eq:two_side_distillable}
\begin{aligned}
  &\norm{U_D U_A\rho_{ABCD} U_A^\dagger   U_D^\dagger  - {\alpha_{A_2 B}} \otimes  \sigma_{A_1 D_1}   \otimes {\beta_{C D_2}} }_1 \\
  &\le  O(\epsilon^{\frac{1}{2}}).
\end{aligned}
\end{equation}
This is the near-maximal distillable entanglement structure on the two sides.
\end{proof}

We note that $\rho_{ABCD}$ and $\sigma_{ABCD}$ can be close in {\it fidelity}, if we allow a unitary transformation on $AD$, i.e. $\exists \,\,U_{AD}$ acting exclusively on $AD$, s.t.
\begin{equation}
\label{eq:U_AD_rho_ABCD}
F( U_{AD} \rho_{ABCD} U_{AD}^\dagger, \sigma_{ABCD} ) \ge \exp( - \epsilon ).
\end{equation}
The comes from the bound on $S_{BC}$ in Eq.~\eqref{eq:S_BC_bound}, which is equivalent to 
\begin{equation}
S\left( \rho_{BC} \Big\| \frac{\I_B}{q} \otimes \frac{\I_C}{q}\right) \le 2\epsilon. 
\end{equation}
Using Theorem~\ref{th:rel_S_F}, we have
\begin{equation}
F\left( \rho_{BC} , \frac{\I_B}{q} \otimes \frac{\I_C}{q}\right)  \ge \exp( -\epsilon ). 
\end{equation}
Since $\sigma_{ABCD}$ is a purification of $\frac{\I_B}{q} \otimes \frac{\I_C}{q}$,  Uhlmann's theorem implies the existence of  $U_{AD}$ satisfying Eq.~\eqref{eq:U_AD_rho_ABCD}.

This alternate version of the input decoupling condition has the disadvantage of obscuring the fact that the $BC$-$AD$ entanglement is mostly due to $B-A$ and $C-D$ entanglement.  However, \Cref{eq:U_AD_rho_ABCD} is stronger for large values of $\eps$ and remains nontrivial when $\eps>1$.  This is relevant in the case of high-dimensional Haar-random unitaries.  It is reasonable to conjecture that \Cref{eq:U_AD_rho_ABCD} holds for some product unitary $U_{AD} = V_A \otimes W_D$.

\subsection{Proof of the almost maximally mixed output state $\rho_{AB'}$}

We complete the proof of \Cref{le:rho_ABp} which states that
the output state $\rho_{AB'}$ is almost maximally mixed.

\begin{proof}
First, from the fidelity bound in Eq.~\eqref{eq:rho_ABp_rhoA_I}, we can decompose $\rho_{AB'}$ with a small trace distance error (\appref{app:dist_mea}),
\begin{equation}
\label{eq:rho_ABp_1}
\norm{ \rho_{AB'} - \rho_{A} \otimes \I /q}_1 \le O( \epsilon^{\frac{1}{2}} ). 
\end{equation}
Then, by taking a partial trace on $BCD$ in the near-maximal distillable entanglement structure in (\Cref{le:two-epr-pair} of the main text), we can approximate $\rho_A$ as 
\begin{equation}
\label{eq:rho_A_decomp}
\norm{\rho_A  - \sigma_{A_1} \otimes\I / q }_1 \le O( \epsilon^{\frac{1}{2}} )
\end{equation}
where $\sigma_{A_1}$ is the reduced density matrix constructed from state $\sigma_{ABCD}$. Replacing $\rho_A$ in Eq.~\eqref{eq:rho_ABp_1} by approximation in Eq.~\eqref{eq:rho_A_decomp} using triangle inequality completes the proof. 
\end{proof}

\section{Discussion of the distance measures}
\label{app:dist_mea} 

In \appref{app:char_max_ent}, we prove the near-maximal distillable entanglement structure of the input state under the assumption of an almost maximal entanglement growth (Eq.~\eqref{eq:app_max_ent_cond}). The proof strategy is to first derive various entropy bounds, e.g. $I(AB; C) \le \epsilon$, and then translate that into the closeness of states. 

From the entropy bounds to the closeness of states, we repeatedly use the following result
\begin{theorem}
\label{th:rel_S_F}
For states $\rho,\sigma$:
\begin{equation}
  F( \rho , \sigma )\ge \exp( - \frac{1}{2} S( \rho \| \sigma ) )
\end{equation}
\end{theorem}
When $S(\rho\|\sigma)$ is small, the fidelity $F( \rho, \sigma ) = \tr\sqrt{\rho^{\frac{1}{2}} \sigma \rho^{\frac{1}{2}}}$ is lower bounded by a number close to $1$, indicating that the two states $\rho$ and $\sigma$ are close. From the general relation 
\begin{equation}
1-F \le D_{\rm tr}  \le \sqrt{ 1 - F^2}, 
\end{equation}
between the fidelity and trace distance $D_{\rm tr}( \rho, \sigma ) = \frac{1}{2} \norm{\rho - \sigma}_1$, we also have that their trace distance is small
\begin{equation}
\frac 12 \norm{\rho- \sigma }_1 \le  \sqrt{1 - \exp( -S(\rho\|\sigma ))}
\le  \sqrt{ S(\rho\|\sigma) }.
\end{equation}
Proceeding from relative entropy to trace distance via fidelity is
slightly less efficient than using Pinsker's inequality, which states that:
\begin{equation}
\frac 12 \norm{\rho- \sigma }_1 \le \sqrt{\frac{S( \rho || \sigma )}{2}}.
\end{equation}
However when $S(\rho\|\sigma)$ is not small (e.g.~$>1$), the bound in Theorem~\ref{th:rel_S_F} is much better. 

Theorem~\ref{th:rel_S_F} is a direct consequence of the monotonicity of the sandwiched R\'enyi divergence\cite{wilde_strong_2014,muller-lennert_quantum_2013,beigi_sandwiched_2013} $\tilde{S}_{\alpha} (\rho || \sigma )$, which is defined to be
\begin{equation}
\tilde{S}_{\alpha} (\rho || \sigma ) = \frac{1}{\alpha - 1} \ln \tr ( ( \sigma^{\frac{1 - \alpha}{2\alpha}} \rho\sigma^{\frac{1 - \alpha}{2\alpha}} )^\alpha ). 
\end{equation}
Its limits at $\alpha = 1$ and $\frac{1}{2}$ relate the relative entropy and the fidelity\cite{muller-lennert_quantum_2013}: 
\begin{equation}\label{eq:sandwiched-limits}
\begin{aligned}
  \lim_{\alpha \rightarrow 1} &\tilde{S}_{\alpha} (\rho || \sigma )  = S( \rho || \sigma ) \\ 
  \tilde{S}_{\frac{1}{2}} &(\rho || \sigma ) = - 2 \ln F( \rho, \sigma ).
\end{aligned}
\end{equation}
Further, $\tilde{S}_{\alpha} (\rho || \sigma )$ is monotonically
increasing in
$\alpha$\cite{muller-lennert_quantum_2013,beigi_sandwiched_2013}.
Together with \cref{eq:sandwiched-limits}, this proves Theorem~\ref{th:rel_S_F}. 

\section{Proximity to a dual unitary at $q = 2$}
\label{app:closeness_q_2}
In the main text, we have proven 
\begin{equation}
\label{eq:app_du_cond}
\norm{\frac{1}{q}\idlineu \otimes \frac{1}{q} \idlineu - \frac{1}{q^2}\uucontract[][tr][epr][tr] }_1 \le O( \epsilon^{\frac{1}{2}} )
\end{equation}
assuming the entanglement growth to be $2 \ln q - \epsilon$. 

A further question is whether $\exists$ a $q^2 \times q^2$ dual unitary $u^{\times}$, such that $\norm{u - u^{\times} }_1 $ is small. In general, we anticipate the error to be of order $\epsilon^{\frac{1}{4}}$, since we are given the error $\epsilon^{\frac{1}{2}}$ in the density matrix, and a version of the decoupling theorem~\cite{devetak_resource_2008,abeyesinghe_mother_2009} tells us the error will take a square root after recovering to, in this case, the operator state. 

We can not directly apply the decoupling trick here, since the result of the recovery should simultaneously satisfy the unitary and dual unitary conditions. There is an iterative algorithm to produce a dual unitary from an arbitrary unitary\cite{rather_creating_2020}, however a rigorous proof of its convergence and if so its rate are lacking. 

Nevertheless we can have the following general weak bound by the compactness of the unitary group: 
\begin{theorem}
Suppose $g$ is a non-negative continuous function on a compact metric space $V$. $\exists$ a non-empty kernel set $S \subset V$, s.t. $\forall u \in S$, $g(u) = 0$. Let 
\begin{equation}
f( x ) = \max \{ \text{\rm dist}( u, S ): g( u) \le x \} \quad \text{ for } x \ge 0 
\end{equation}
then
\begin{equation}
\lim_{x \rightarrow 0} f(x) = 0 
\end{equation}
\end{theorem}
In the context of our problem, the non-negative continuous function $g$ is the LHS of Eq.~\eqref{eq:app_du_cond}, the compact space $V$ is the unitary group $\text{U}(q^2)$, and the set $S$ is the space of dual unitaries. The function $f(x)$ characterizes the maximal distance (can be any distance measure) a unitary can have, when the dual unitary condition is approximately satisfied to the extent of Eq.~\eqref{eq:app_du_cond}. The theorem tells us that when gradually tightening the approximation of the dual unitary condition, we will approach the dual unitary space. 
\begin{proof}
  First we note that the $\max$ in $f$ always exists.  This follows
  from the fact that $\text{dist}$ is continuous and $g^{-1}([0,x])$ is
  compact.


Next, define a sequence $(u_n)_n$ by choosing $u_n$ arbitrarily from
the points achieving the maximum in $f(1/n)$.
\begin{equation}
u_n \in \text{\rm argmax}\{ \text{dist}( u, S ) , g(u) \le 1/n  \}. 
\end{equation}
(Here $1/n$ could be any sequence converging to 0.)
 By definition
\begin{equation}
f( 1/n  ) = \text{\rm dist}( u_n , S ). 
\end{equation}
$u_n$ is a sequence in a compact space, hence there is a subsequence $u_{n_i}$ so that $u_{n_i}$ converges to a fixed point $v$. Since $g( u_{n_i} ) \le 1/{n_{i}}$ and $g$ is continuous,  we have $g( v ) = \lim_{i\rightarrow \infty} g(u_{n_i} ) \le \lim_{i\rightarrow \infty} 1/{n_i} = 0$. This means $g(v) = 0$ and $v \in S$. Since the distance function is also continuous, we have
\begin{equation}
  \lim_{i\rightarrow \infty} f( 1/{n_i} )
  = \lim_{i\rightarrow \infty} \text{dist} (u_{n_i}, S )
  = \text{dist} (v, S )  = 0 
\end{equation}
Finally since $f$ is a nonincreasing and nonnegative function, we
conclude that $f(x)$ is continuous at $x=0$.
\end{proof}

We may make further progress if we can parameterize the dual unitary manifold. Currently, a {\it general} parameterization is still lacking\cite{prosen_many_2021}, but we have an explicit one for $q = 2$. 
\begin{claim}
 A $4 \times 4$ unitary has a Cartan decomposition\cite{khaneja_cartan_2000}
\begin{equation}
\label{eq:sufour}
u = e^{i\phi} \sufour[u_{\rm sym}][u_{1}][u_{2}][u_3][u_4]
\end{equation}
where the symmetric part is generated by the Cartan algebra
\begin{equation}
\label{eq:u_sym}
u_{\rm sym} = \exp \left( -i \sum_{\alpha = x, y,z} J_{\alpha}\, \sigma^\alpha \otimes \sigma^\alpha  \right)
\end{equation}
and $u_{1,2,3,4} \in {\rm SU}(2)$. If $u$ is dual unitary, then at least two of $J_x$, $J_y$, $J_z$ have absolute value $\frac{\pi}{4}$\cite{bertini_exact_2019-1,gopalakrishnan_unitary_2019}. 
\end{claim}

This helps us to prove the following:
\begin{theorem}
If $q = 2$, and 
\begin{equation}
\label{eq:close_to_dual}
\norm{\idlineu \otimes  \idlineu - \uucontract[][tr][epr][tr] }_1 \le \delta 
\end{equation}

then, $\exists$ $u^{\times}  \in $ $4 \times 4$ dual unitary, s.t. 
\begin{equation}
\norm{u - u^{\times} }_1 \le O( \sqrt{ \delta } )
\end{equation}
\end{theorem}

\begin{proof}
We parameterize the unitary gate $u$ by Eq.~\eqref{eq:sufour}. Since 1-norm of the matrix is unitarily invariant, Eq.~\eqref{eq:close_to_dual} also holds for its symmetric part $u_{\rm sym}$.
\begin{equation}
\label{eq:sym_close_to_dual}
\norm{\idlineu \otimes  \idlineu - \uucontract[][tr][epr][tr][\text{sym}] }_1 \le \delta 
\end{equation}
We take $u^{\times}$ to have the same $u_1$, $u_2$, $u_3$ and $u_4$ as in the parameterization of $u$, but take a different symmetric part $u_{\rm sym}^{\times}$ such that it is a dual unitary. Then
\begin{equation}
\norm{ u - u^{\times}}_1 = \norm{ u_{\rm sym} - u^{\times}_{\rm sym} }_1 
\end{equation}
We only need to designate the symmetric part. We set
\begin{equation}
u = \sum_{\alpha = \text{id}, x, y, z}C_\alpha \sigma^\alpha \otimes \sigma^\alpha
\end{equation}
Eq.~\eqref{eq:sym_close_to_dual} becomes
\begin{equation}
  \norm{ (2 |C_{\alpha}|^2  - \frac{1}{2} ) \sigma^\alpha \otimes \sigma^{\alpha *} }_1 \le \delta 
\end{equation}
The 1-norm is computed to be
\begin{equation}
\begin{aligned}
&| |C_0|^2 - |C_1|^2 - |C_2|^2+ |C_3|^2   | \\
&+ | |C_0|^2 - |C_1|^2 + |C_2|^2 - |C_3|^2   | \\
&+ | |C_0|^2 + |C_1|^2 - |C_2|^2 - |C_3|^2   |  \\
=& | \cos 2 J_x \cos 2 J_y | + | \cos 2 J_y \cos 2 J_z | + | \cos 2 J_x \cos 2 J_z |  
\end{aligned}
\end{equation}
Since it is no greater than $\delta$, we have any one of $| \cos 2 J_x \cos 2 J_y | $, $ | \cos 2 J_y \cos 2 J_z | $,$ | \cos 2 J_x \cos 2 J_z | $ to be no greater than $\delta$. 
Among $|\cos 2 J_x |$, $|\cos 2 J_y |$, $|\cos 2 J_z |$, at most one of them can be $> \sqrt{ \delta}$. Without loss of generality, we assume
\begin{equation}
|\cos 2 J_x | \le \sqrt{ \delta }, \quad |\cos 2 J_y | \le \sqrt{ \delta }
\end{equation}
For sufficiently small $\delta$, $J_x$ can confined in one of the intervals
\begin{equation}
\left|J_x - \frac{\pi}{4}\right| \le \frac{1}{2} \arcsin{\sqrt{ \delta }} < \sqrt{ \delta}  \quad \left|J_x + \frac{\pi}{4}\right| \le \sqrt{ \delta }
\end{equation}
Without loss of generality, we assume $J_x$ and $J_y$ are close to $\frac{\pi}{4}$. The analysis for the $- \frac{\pi}{4}$ case is completely the same. 

Then we choose $u^{\times}_{\rm sym}$ to be parameterized by $J^{\times}_x = \frac{\pi}{4} $, $J^{\times}_y = \frac{\pi}{4}$, and $J^{\times}_z = J_z$. The bottom line is, we can always choose
\begin{equation}
\left| e^{i J_x^{\times} } - e^{i J_x } \right| \le |J_x - J_x^{\times} | \le \sqrt{\delta }, \quad \left| e^{i J_y^{\times} } - e^{i J_y } \right|  \le \sqrt{\delta} 
\end{equation}
Finally, we control the distance between the unitaries:
\begin{equation}
\norm{u - u^{\times} }_1 = \norm{\sum_\alpha ( C_\alpha - C_\alpha^{\times}) \sigma^\alpha \otimes \sigma^\alpha }_1
\end{equation}
This $1$-norm is computed to be
\begin{equation}
\begin{aligned}
&\norm{u - u^{\times} }_1 =\,\, \left| e^{ i ( J_x + J_y + J_z) } - e^{ i ( J^{\times}_x + J^{\times}_y + J^{\times}_z) }  \right|  \\
&+ \Big| e^{ i ( J_x + J_y + J_z) } - e^{ i ( J^{\times}_x + J^{\times}_y + J^{\times}_z) } \\
  &\qquad +   e^{ i ( J_x - J_y - J_z) } - e^{ i ( J^{\times}_x - J^{\times}_y - J^{\times}_z) } \Big| \\
&+ \Big| e^{ i ( J_x + J_y + J_z) } - e^{ i ( J^{\times}_x + J^{\times}_y + J^{\times}_z) } \\
  &\qquad +   e^{ i ( -J_x + J_y - J_z) } - e^{ i ( -J^{\times}_x + J^{\times}_y - J^{\times}_z) } \Big|  \\
&+ \Big| e^{ i ( J_x + J_y + J_z) } - e^{ i ( J^{\times}_x + J^{\times}_y + J^{\times}_z) } \\
  &\qquad +   e^{ i ( -J_x - J_y + J_z) } - e^{ i ( -J^{\times}_x - J^{\times}_y + J^{\times}_z) } \Big|
\end{aligned}
\end{equation}
Let $s_{x,y,z} = \pm 1$,
\begin{equation}
\begin{aligned}
&\left| e^{ i \sum_{i=x,y,z}s_{i} J_i  } - e^{ i \sum_{i=x,y,z} s_i J_i^{\times} }\right| \\
& = \left| e^{ i \sum_{i=x,y}s_{i} J_i  } - e^{ i \sum_{i=x,y} s_i J_i^{\times} }\right|\\  
&\le  \left|( e^{ i s_x J_x}  - e^{is_x J_x^{\times} }) e^{ is_y J_y } \right| 
+ \left| e^{ i s_x J_x^{\times} } ( e^{i s_y J_y } - e^{i s_y J^{\times}_y } ) \right| \\
&\le 2 \sqrt{\delta } 
\end{aligned}
\end{equation}
Hence by breaking the three terms using the triangle inequality
\begin{equation}
\norm{u - u^{\times} }_1 \le 14 \sqrt{\delta} 
\end{equation}
\end{proof}

\section{The separating states for the kicked Ising model}
\label{app:sep_states_in_skdi}
The kicked Ising model has a unitary evolution with a Floquet operator
\begin{equation}
U = \exp( - i H_K ) \exp( -i H_I ) 
\end{equation}
where $H_I$ is the Ising Hamiltonian
\begin{equation}
H_I = J \sum_{j=0}^{2L-1} Z_j Z_{j+1} + h \sum_{j=0}^{2L-1} Z_j
\end{equation}
and $H_K$ is the periodic kick
\begin{equation}
H_K = b \sum_{j=0}^{2L-1} X_j.
\end{equation}
\begin{figure}[h]
\centering
\includegraphics[width=0.75\columnwidth]{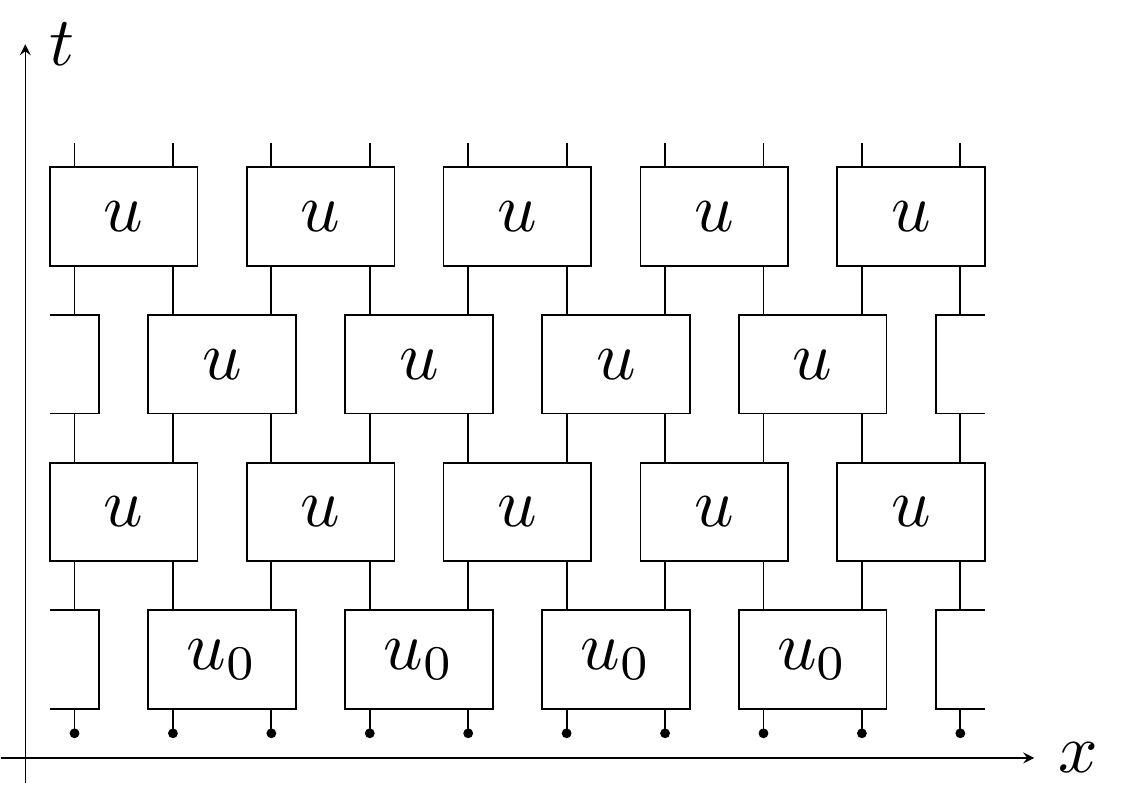}
\caption{The kicked Ising model in circuit representation. The first layer has different unitary gates than other layers. }
\label{fig:sdki}
\end{figure}

It has a circuit representation, where the gate for $t \ge 2$ assume the form\cite{bertini_exact_2019-1,gopalakrishnan_unitary_2019}
\begin{equation}
\begin{aligned}
u = &[\exp( -i h Z / 2 ) \otimes \exp( -i h Z / 2 )]  \\
    & \qquad \exp( - i J ZZ )  \\
    & [ \exp( -i b X ) \otimes \exp( -i b X )]\\
    & \qquad \exp( - i J ZZ ) \\
    & [\exp( -i h Z / 2) \otimes \exp( -i h Z  / 2 )]  
\end{aligned}
\end{equation}
while at the first layer the gate is taken to be
\begin{equation}
u_0 = \exp( -i J Z_1 Z_2 - i h Z_1 / 2 - i h Z_2 / 2 ). 
\end{equation}
At the self-dual (dual unitary) point, $|J| = |b| = \frac{\pi}{4}$. We take both parameters to be positive. 

There are two classes of separating states in Ref.~\onlinecite{bertini_entanglement_2018}:
\begin{enumerate}
\item $\mathcal{T}$: tensor product of spin states on the $xy$ plane.
\item $\mathcal{L}$: tensor product of spin states in $z$ direction. 
\end{enumerate}

Given the $\mathcal{T}$ class state, the interaction $\exp( -i
\frac{\pi}{4} Z_1 Z_2 ) $ $= \frac{\sqrt{2}}{2} ( \I - i Z_1 Z_2)$ in
$u_0$ creates a Bell state on the two qubits it acts on. Then the
zigzag pattern is formed, and any dual unitary gate in the circuit afterwards can create maximal entanglement growth. 

Given the $\mathcal{L}$ class state, the $u_0$ layer only generates
phases. The gate $u$ can be written in the standard form in Eq.~\eqref{eq:sufour} with $u_1 = u_2  = \exp( -i h Z / 2)$, and $J_y = J_z = \frac{\pi}{4}$. Again, $u_1 $, $u_2$  only creates phases. The symmetric part generates a Bell pair. $u_3$, $u_4$, whatever they are, do not change the entanglement. Thus the zigzag pattern is exactly established at $t = 2$, and the maximal entanglement growth can resume afterwards.

\section{The solvable MPS ansatz satisfies the zigzag pattern}
\label{app:smps}
In addition to the separating states in the kicked Ising model,
Ref.~\cite{piroli_exact_2020} introduces a wider class of ``solvable''
matrix product states for the dual unitary circuit.  These states are
invariant under a two-site shift and can be written as:
\begin{equation}
| \psi \rangle = \sum_{i_1, i_2, \cdots, i_n} \cdots A^{i_1} B^{i_2} A^{i_3} B^{i_4} \cdots | \cdots i_1 i_2 i_3 \cdots \rangle. 
\end{equation}
The state can be diagrammatically represented as
\begin{equation}
\cdots 
\fineq[-3ex][0.5][1]{
  \mpsa[0][0][][$A$][][$i_1$][];
  \mpsa[2][0][][$B$][][$i_2$][];
  \mpsa[4][0][][$A$][][$i_3$][];
  \mpsa[6][0][][$B$][][$i_4$][];
}
\cdots 
\end{equation}
The dimension of the physical index (vertical) is $q$ and the dimension of the auxiliary index is $\chi$ from $B$ to $A$ and $\chi'$ from $A$ to $B$. We will use open boundary condition and take the system size $L \rightarrow \infty$ limit in the end. The state is solvable when the tensor
\begin{equation}
\mathcal{N}_{ab}^{ij} \equiv 
\sqrt{q} \times 
\fineq[-3ex][0.5][1]{
  \mpsa[0][0][][$A$][$a$][$i$][];
  \mpsa[2][0][][$B$][][$j$][$b$];
}
\end{equation}
is a unitary matrix with $a$, $i$ as its row indices and $b$, $j$ as its column indices. This indicates that $\chi' = \chi q$. 

We denote the entanglement along the cut between $A$ and $B$ as $E(A:B)$, and long the cut between $B$ and $A$ as $E(B:A)$. We will show that
\begin{equation}
\label{eq:S_A_cut_B}
E(A:B) = \ln \chi + \ln q .
\end{equation}
On the other hand, we have
\begin{equation}
E(B:A) \le \ln \chi 
\end{equation}
and
\begin{equation}
E(B:A) \ge E(A:B) - \ln q = \ln \chi.  
\end{equation}
Hence
\begin{equation}
E(B:A) = \ln \chi.  
\end{equation}
The difference of the entanglement across neighboring bounds are exactly $\ln q$. Therefore the solvable MPS ansatz satisfies the zigzag pattern.

Non we follow the strategy in Ref.~\cite{piroli_exact_2020} and use a replica trick to show Eq.~\eqref{eq:S_A_cut_B}. For simplicity we only show $\tr(\rho_Q^2)$. It is
\begin{equation}
\tr( \rho_Q^2 ) = 
\cdots
\fineq[-3ex][0.5][1]{
  \mpsastack[0][0][][$B$][][][];
  \idarc[0][1];
  \mpsastack[2][0][][$A$][][][];
  \idarc[2][1];
  \mpsastack[4][0][][$B$][][][];
  \swaparc[4][1];
  \mpsastack[6][0][][$A$][][][];
  \swaparc[6][1];
}
\cdots 
\end{equation}
The definition of the solvable MPS ensures that the there is a unique left eigenvector and a unique right eigenvector of the transfer matrix of the MPS
\begin{equation}
\fineq[-3ex][0.5][1]{
\mpsa[-0.1][0.1][blue!50][];
\mpsa[0][0][red!50][$A$];
\shortarc[0][1];
\begin{scope}[shift={(2,0)}]
\mpsa[-0.1][0.1][blue!50][];
\mpsa[0][0][red!50][$B$];
\shortarc[0][1];
\end{scope}
}
\end{equation}
whose eigenvalue is exactly $1$. As diagrams, these two eigenvalues are
\begin{equation}
| R \rangle  =  \frac{1}{\sqrt{ \chi} }
\fineq[-0.5ex][2][1]{
\begin{scope}[yscale = -1, rotate=-90]
\shortarc[0][0];
\end{scope}
},
\quad
\langle  L | = 
\frac{1}{\sqrt{ \chi} }
\fineq[-0.5ex][2][1]{
\begin{scope}[yscale = -1, rotate=90]
\shortarc[0][0];
\end{scope}
}.
\end{equation}
By the normalization of the MPS, their overlaps with the left and right boundary conditions are $1$. We can thus replace the ellipsis with tensor products of those eigenvectors and obtain
\begin{equation}
\tr(\rho_Q^2)
= 
\frac{1}{\chi^2 q^2 }\times \left( q^2
\fineq[-2ex][0.5][1]{
  \begin{scope}[yscale=-1,xscale=1,rotate=90]
    \idarc[0][1];
  \end{scope}
  \mpsastack[0][0][][$A$][][][];
  \idarc[0][1];
  \mpsastack[2][0][][$B$][][][];
  \swaparc[2][1];
  \begin{scope}[yscale=-1,xscale=1,rotate=-90]
    \swaparc[0.3][2.7];
  \end{scope} 
}\right).
\end{equation}
This is a the purity of the operator state defined by the unitary $\mathcal{N}_{ab}^{ij}$ when the entanglement cut is imposed between its row and column indices. As a result
\begin{equation}
\tr( \rho_Q^2 ) = \frac{1}{ \chi q }. 
\end{equation}
By using the same diagrammatic evaluation, we can similarly show that
\begin{equation}
\tr( \rho_Q^n) = \frac{1}{(\chi q)^{n-1}},  
\end{equation}
which means
\begin{equation}
E(A:B) = S(Q) = \ln \chi + \ln q. 
\end{equation}

\section{Fidelity of a random unitary to a dual unitary}
\label{app:ru_fidelity}
How close is a randomly chosen unitary to a dual unitary? In the main text (Theorem 1), we use the trace distance of the density matrix $\rho_{AB'}$ with respect to the maximally mixed state to characterize the proximity for a general choice of $q$. Here we evaluate the fidelity
\begin{equation}
\begin{aligned}
  F( \rho_{AB'}, \frac{\I_{AB'} }{q^2} )
  &= \tr( \sqrt{\frac{1}{q} \rho_{AB'} \frac{1}{q}} ) = \frac{1}{q } \tr( \sqrt{\rho_{AB'}} ),      
\end{aligned}
\end{equation}
assuming the input state $|\psi_{AB} \rangle$, $|\psi_{CD} \rangle$
are Bell states, and the unitary $u$ is randomly
chosen. Alternatively, one can think of $\rho_{AB'}$ as tracing out
$C'D$ from the Choi-Jamio\l{}kowski (or operator) state $\ket{u}_{AB'C'D}$.

The random average of such (generalized) purity at integer $n$ can be evaluated (see e.g. \cite{zhou_operator_2017,nadal_statistical_2011-1})
\begin{equation}
\overline{\text{tr}(\rho_{AB'}^n)} = \frac{C_n}{q^{2(n-1)}} + \mathcal{O}(\frac{1}{q^{2n+2}})
\end{equation}
where $C_n$ is the Catalan number
\begin{equation}
C_n = \frac{(2n)!}{n! (n+1)!} 
\end{equation}

Analytically continuing to $n = \frac{1}{2}$ , we get
\begin{equation}
\begin{aligned}
  \overline{F( \rho_{AB'}, \I_{AB'} /q^2 ) } &= \frac{8}{3\pi} + \mathcal{O}(\frac{1}{q^2})   \approx 0.8488 + \mathcal{O}(\frac{1}{q^2})
\end{aligned}
\end{equation}
Therefore, a randomly chosen unitary is close be a dual unitary with a fidelity measure about $0.8488$ for the dual unitary condition. 

This is analogous to the approximate unitary property of a random
quantum state. Specifically, a quantum state $|\psi_{AB} \rangle $ on
two qudits $A$ and $B$ can be viewed as an operator from $A$ to $B$
via the Choi-Jamio\l{}kowski isomorphism.  This operator will be
unitary iff $\rho_A = \I_A / q$. For a random state, the entanglement
$S(A)$ is $\ln(q) - \mathcal{O}(1)$ when $q$ is large, which means
that $\rho_A $ for a random state has $O(1)$ relative entropy distance to $\I_A / q$. We estimate the following fidelity
\begin{equation}
\begin{aligned}
F( \rho_{A}, \I_{A} /q ) &= \tr( \sqrt{\frac{1}{\sqrt{q}} \rho_{A} \frac{1}{\sqrt{q}}} )  = \frac{1}{\sqrt{q} } \tr( \sqrt{\rho_{A}} ). 
\end{aligned}
\end{equation}
Analytically continuing the random average result (see e.g. \cite{zhou_operator_2017,nadal_statistical_2011-1}) 
\begin{equation}
\overline{\text{tr}(\rho_{A}^n)} = \frac{C_n}{q^{(n-1)}} +  \mathcal{O}(\frac{1}{q^{n+1}})
\end{equation}
to $n = \frac{1}{2}$, we also have
\begin{equation}
\overline{F( \rho_{A}, \I_{A} /q ) } = \frac{8}{3\pi} + \mathcal{O}(\frac{1}{q^2})   \approx 0.8488 + \mathcal{O}(\frac{1}{q^2}). 
\end{equation}

\bibliographystyle{apsrev4-1}
\bibliography{max_vE_paper}

\begin{thebibliography}{51}
\bibitem [{FN()}]{FN}%
  \BibitemOpen
  \href@noop {} {}\bibinfo {note} {This is weaker than requiring the R\'enyi
  entanglement velocity (with R\'enyi index $n \ge 1$) to be equal to $1$, as
  von Neumann entropy is always greater than or equal to R\'enyi entropies with
  a positive integer index.}\BibitemShut {Stop}%
\end{thebibliography}

\end{document}